\documentclass[10pt,journal,compsoc]{IEEEtran}
\ifCLASSOPTIONcompsoc
  \usepackage[nocompress]{cite}
\else
  \usepackage{cite}
\fi
\usepackage{xcolor}
\usepackage{float}
\usepackage{amsthm}
\usepackage{graphicx}
\newtheorem{proposition}{Proposition}
\newtheorem{observation}{Observation}
\ifCLASSINFOpdf
\else
\fi
\usepackage{amsmath}
\usepackage{booktabs}
\usepackage{multirow}
\usepackage{makecell}
\usepackage{algorithm2e}[ruled]
\RestyleAlgo{ruled}
\usepackage{amsfonts}

\begin{document}
\title{UniqueRank: Identifying Important and Difficult-to-Replace Nodes in Attributed Graphs}

\author{Erica~Cai,
Benjamin~A.~Miller,
        Olga~Simek,
        and~Christopher~L.~Smith% <-this % stops a space
\IEEEcompsocitemizethanks{\IEEEcompsocthanksitem B. A. Miller, O. Simek, and C. L. Smith are with Lincoln Laboratory, Massachusetts Institute of Technology, Lexington, MA, 02421.\protect\\
E-mail: \{bamiller, osimek, clsmith\}@ll.mit.edu 
\IEEEcompsocthanksitem E. Cai is currently with Meta. All work was performed during her time at Lincoln Laboratory.\protect\\
E-mail: ecai@umass.edu}% <-this % stops an unwanted space
\thanks{This material is based upon work supported by the Combatant Commands under Air Force Contract No. FA8702-15-D-0001 or FA8702-25-D-B002. Any opinions, findings, conclusions or recommendations expressed in this material are those of the author(s) and do not necessarily reflect the views of the Combatant Commands.}}


\markboth{In Submission}%
{Shell \MakeLowercase{\textit{et al.}}: Bare Demo of IEEEtran.cls for Computer Society Journals}



\IEEEtitleabstractindextext{%
\begin{abstract}
Node-ranking methods that focus on structural importance are widely used in a variety of applications, from ranking webpages in search engines to identifying key molecules in biomolecular networks. In real social, supply chain, and terrorist networks, one definition of importance considers the impact on information flow or network productivity when a given node is removed. In practice, however, a nearby node may be able to replace another node upon removal, allowing the network to continue functioning as before. This \emph{replaceability} is an aspect that existing ranking methods do not consider. To address this, we introduce UniqueRank, a Markov-Chain-based approach that captures attribute uniqueness in addition to structural importance, making top-ranked nodes harder to replace. We find that UniqueRank identifies important nodes with dissimilar attributes from its neighbors in simple symmetric networks with known ground truth. Further, on real terrorist, social, and supply chain networks, we demonstrate that removing and attempting to replace top UniqueRank nodes often yields larger efficiency reductions than removing and attempting to replace top nodes ranked by competing methods. Finally, we show UniqueRank's versatility by demonstrating its potential to identify structurally critical atoms with unique chemical environments in biomolecular structures.
\end{abstract}

}

\maketitle

\IEEEdisplaynontitleabstractindextext

% \fi
\IEEEpeerreviewmaketitle

\IEEEraisesectionheading{\section{Introduction}\label{sec:introduction}}


\IEEEPARstart{I}{dentifying} important nodes within a graph is a critical task across a wide range of real-world applications. In terrorist networks, a few key individuals can orchestrate large-scale operations; in supply chains, certain warehouses are crucial for the efficient distribution of essential goods; and on social media, influential users can help to rapidly spread information. Numerous methods have been developed to rank the importance of nodes~\cite{ZAREIE2018hierarchical,WANG2017ranking,cheng2011virtual,SALAVATI2019ranking}, with PageRank~\cite{Page1998PageRank,brin1998anatomy} being one of the most widely used. PageRank uses a Markov Chain model based on the assumption that high-ranking nodes are those that are linked to by other high-ranking nodes. More recent approaches, such as AttriRank~\cite{hsu2017unsupervised} and similar methods~\cite{benyahia2015centrality}, extend these methods by integrating node attributes, such as organizational membership, into the ranking process, which captures richer contextual information.

A persistent challenge in real-world applications is that the removal of important nodes can greatly compromise a network’s ability to perform its intended function. However, many methods for identifying critical nodes overlook the fact that networks are often resilient and can adapt to disruptions by substituting the removed nodes with alternative nodes located nearby~\cite{missaoui2013social}. For example, in a terrorist organization, if a central leader is removed, another operative with similar capabilities may step into the leadership role, with network connections subsequently rerouting to preserve the operational structure. In such cases, while the initial removal may temporarily destabilize the system, the long-term impact on overall effectiveness may be minimal if a suitable replacement is readily available. Therefore, the extent to which node removal and subsequent replacement disrupts network efficiency depends not only on the structural importance of the node but also on the presence and accessibility of potential replacements within the network.

\begin{figure}
\centering
    \includegraphics[width=\columnwidth]{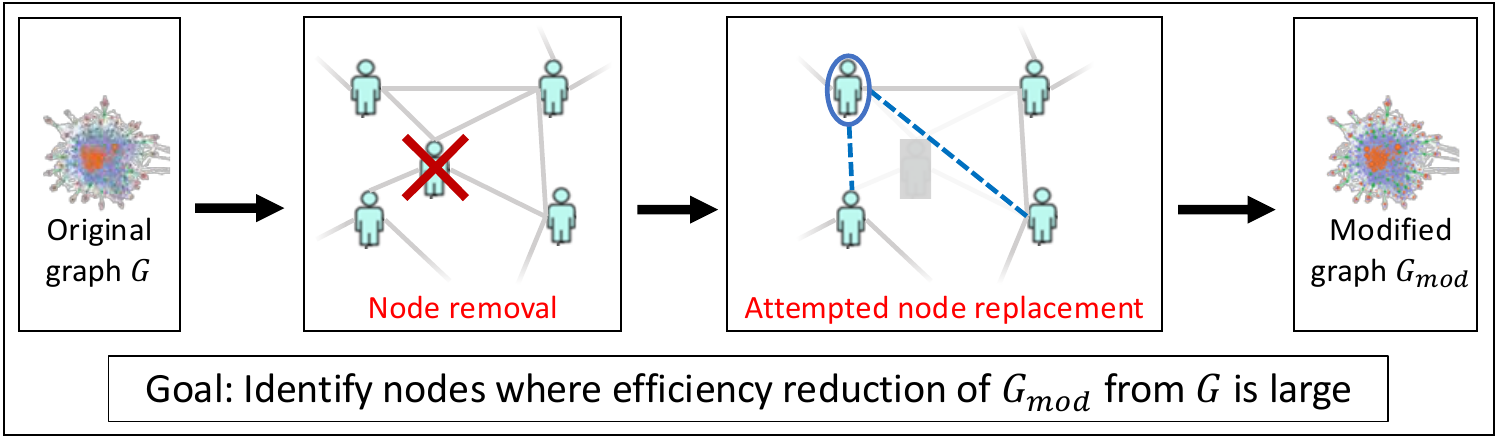}
    \caption{Pipeline overview for evaluating network efficiency reduction after node removal and attempted replacement. The process involves identifying a key node, removing it, searching for a suitable replacement, and then measuring the resulting change in efficiency.}
    \label{fig:goal}
\end{figure}

To address this gap, we focus on identifying nodes whose removal and subsequent attempted replacement result in a significant decrease in the network's overall efficiency\cite{latora2001efficient}, as in Figure \ref{fig:goal}. Specifically, we propose a novel node ranking task that prioritizes nodes that are both structurally important and have attributes that are distinct from those in their immediate neighborhood. By emphasizing attribute uniqueness in addition to structural importance, our approach aims to pinpoint nodes that are not only critical to the network's functioning but also difficult to replace effectively. Although numerous methods rank nodes by structural importance---such as traditional centrality metrics~\cite{Rodrigues2019network,Das2018-aj}, Markov Chain-based approaches~\cite{Page1998PageRank,brin1998anatomy}, and newer graph neural network models~\cite{QU2023gnr,maurya2021graph,li2024novel,park2019estimating,ergashev2023learning,he2022gnnrank}---few incorporate additional node attribute information. In most cases, attributes are derived directly from topology, while in cases with extraneous node attributes~\cite{benyahia2015centrality}, uniqueness is not generally considered. Furthermore, to our knowledge, no existing method specifically ranks nodes by both structural importance and attribute uniqueness. Naive approaches that consider importance and uniqueness as separate variables in a single function typically require customized parameters for each graph.

We introduce UniqueRank, a provably efficient method that builds on AttriRank and incorporates node attribute information into a Markov Chain model to rank nodes by structural importance and uniqueness of attributes. UniqueRank features a tunable hyperparameter that allows users to balance the emphasis between these two criteria. We demonstrate that, in symmetric networks with known ground truth, UniqueRank reliably assigns higher ranks to nodes that are both structurally central and possess distinctive attributes. Therefore, the top UniqueRank nodes form a robust set of candidates that are both critical and hard to replace. A secondary refinement step is then applied to select the final set of top-ranked nodes that maximize both structural importance and attribute uniqueness.

We rigorously evaluate efficiency reduction of top UniqueRank nodes across diverse domains---real terrorist networks, social networks, and supply chain networks---and find that removing and then attempting to replace these top-ranked nodes consistently leads to larger efficiency reduction in the neighborhood of the removed node when compared to other widely used and traditional ranking methods. We measure efficiency reduction~\cite{latora2001efficient} by quantifying the relative decrease in shortest path distances between node pairs after a top-ranked node is removed and replaced. Our results hold across different attribute-similarity and node-distance criteria for assessing a node's suitability as a replacement for a removed node.

Beyond terrorist, social, and supply chain networks, we demonstrate that identifying structurally important nodes that have unique attributes is useful in other domains, such as the analysis of biomolecule structures. For example, we find that UniqueRank identifies atoms that are both structurally well-connected and have unique chemical environments, yielding unique interactions with other atoms.

In summary, our main contributions are as follows:
\begin{itemize}
\item We propose \textbf{UniqueRank}, a new method for identifying nodes that both are structurally important and have unique attributes.
\item We provide a comprehensive evaluation of UniqueRank on real-world social, terrorist, and supply chain networks, demonstrating that the removal and attempted replacement of top-ranked UniqueRank nodes results in greater efficiency reduction compared to traditional and widely used ranking methods.
\item We show the broader applicability of UniqueRank by applying it to biomolecular structures.
\end{itemize}
%{\color{red}SHOULD WE PUT A PAPER ORGANIZATION PARAGRAPH HERE???? I THINK SO :)}

The remainder of this paper is organized as follows. Section~\ref{sec:related-work} reviews related work on node ranking, network resilience, and the integration of node attributes into importance measures. Section~\ref{sec:uniquerank} formally defines the problem of identifying structurally important nodes with unique attributes and introduces the UniqueRank method, providing theoretical insights and analysis of its properties. Section~\ref{sec:experiments} describes our experimental setup and results, including details on real world social network datasets and evaluation metrics. Finally, Section~\ref{sec:biomolecule} applies UniqueRank to the biomolecular domain.

\section{Related Work} \label{sec:related-work}

\subsection{Node Removal Attacks}

\textbf{Node Removal and Network Disruption.} The consequences of node removal have been extensively explored in network science~\cite{bellingeri2020link}. Several studies have found that random removal of nodes often has limited impact on overall network functionality~\cite{Amancio_2015}, but targeted removal of high-degree nodes is much more likely to cause substantial disruption~\cite{Albert2000-cm,CallawayNewmanStrogatzWatts+2006+510+513,Albert_2002,holme2002attack}. Numerous works have focused on developing attack strategies to identify individual nodes or sequences of nodes whose removal maximizes network disruption~\cite{chen2013node,JAHANPOUR20133458,Boldi2013-wp,Iyer2013-fm,BELLINGERI2014174,gallos2005stability,NIE2015248,Tian2017-en,NGUYEN2019121561,Morone2015-cg,Schneider_2012inverse}. Research on node attacks is particularly useful for social networks, where strategic node removal can provide insights into how to hamper information flow, for problems such as minimizing disease spread by removing highly connected individuals~\cite{P.Holme_2004,Wang_2015immunity,schneider2011suppressing,bellingeri2015optimization,chen2008finding,Hadidjojo2011-equal}, identifying the most influential contributors in citation networks~\cite{Pan_2012strength}, or disrupting criminal operations by targeting key actors in criminal networks~\cite{AGRESTE201630,requiao2018topology,petersen2011node}. Moreover, Qi et al.~\cite{qi2013terrorist} introduced a Laplacian centrality measure that accounts for the reorganization of network paths after node removal, demonstrating that networks tend to adapt more poorly when the highest-ranked nodes are removed.

%{\color{red}THIS LAST SENTENCE DOESN'T MAKE SENSE.}

\vspace{0.4em}
\noindent \textbf{Application-Specific Disruption.} A large amount of literature on network disruption is tailored to particular applications, such as supply chain and transportation networks. Recent work investigated robustness of supply chain networks to disruptions, and the impact of network topology on robustness~\cite{zhao2019supply,Thadakamaila2004survivability,zhao2011analyzing}, as well as methods to construct more robust supply chain networks~\cite{tian2021research,Shi2012ADM,zhao2011achieving,bimpikis2019supply}. In a more specific application,~\cite{Chen2023-ue} investigates the evolution of cold chain agricultural product movement in China following network disruptions during COVID-19.  Various works have explored the effect of transportation network disruptions, such as the I-35W bridge collapse, on travel and traffic behavior~\cite{zhu2012disruptions}. This literature includes investigation of how traffic evolves after a network disruption~\cite{HE2012modeling}. 

\subsection{Node Ranking Methods}

\textbf{Traditional Centrality Metrics.} Ranking nodes by structural importance is a foundational task in network analysis, and has historically relied on core centrality metrics such as degree, closeness, betweenness, and eigenvector centrality~\cite{Rodrigues2019network,Das2018-aj}. Closeness centrality extends degree centrality by considering a node's position within the entire network, rather than only its immediate neighbors, while betweenness centrality captures how influential a node is in facilitating connections between other nodes. However, both closeness and betweenness centrality are based on the assumption that information flows along shortest paths and are computationally expensive to calculate~\cite{Landherr2010-fg}. Eigenvector centrality further captures influence of a node by considering both the number and importance of connected nodes, but is always zero for directed acyclic graphs~\cite{BONACICH2007555,grassi2007some}. To measure the impact of node removal, we use network efficiency~\cite{latora2001efficient}, defined as the sum of the reciprocals of the shortest path lengths between all node pairs. Many existing methods focus on identifying those structurally important nodes whose removal most drastically lowers efficiency~\cite{chen2013node,JAHANPOUR20133458,Boldi2013-wp,Iyer2013-fm,BELLINGERI2014174,gallos2005stability,NIE2015248,Tian2017-en,NGUYEN2019121561,Morone2015-cg,Schneider_2012inverse}. However, these methods often do not consider the outcome after possible node replacement.

\vspace{0.4em}
\noindent \textbf{Extended Ranking Methods.} Many ranking methods have been proposed to improve upon the traditional centrality metrics. HITS is an iterative algorithm that ranks webpages by updating an authority score, based on hyperlinks pointing to each page, and a hub score, based on hyperlinks leading out of a page~\cite{Kleinberg99hits}. A weighted aggregation of centrality metrics presented by~\cite{yang2019node} and~\cite{NEWMAN200539} proposes a random walk method to efficiently approximate betweenness centrality, and~\cite{bandes2001faster} also proposes faster methods. PageRank~\cite{Page1998PageRank,brin1998anatomy} is a Markov Chain--based node ranking method originally used to rank webpages in Google search engine.  %{\color{red}$\leftarrow$TODO: fill in this citation (actually, is everything from ``It assumes . . .'' to the end of this paragraph better put in Section~\ref{sec:uniquerank}?)}

Recent importance ranking methods use deep learning, such as graph neural networks~\cite{QU2023gnr,maurya2021graph,li2024novel,park2019estimating,ergashev2023learning,he2022gnnrank}. Others consider dynamic graphs~\cite{chen2021node}, leverage entropy-based measures~\cite{Liu2023structure,Yu2022-identifying} and other graph structural information~\cite{ZAREIE2018hierarchical,WANG2017ranking,cheng2011virtual,SALAVATI2019ranking} to inform node rankings. 

\vspace{0.4em}
\noindent \textbf{Ranking with External Node Attributes.} While the methods above aim to measure structural importance of nodes, many real-world networks are also characterized by node attributes that encode additional, often domain-specific, information (e.g., organizational memberships). Few existing approaches consider such attributes when ranking nodes by importance. Most methods that do incorporate attributes rely on structural metrics such as h-index---defined as the maximum integer $h$ such that at least $h$ nodes have degree at least $h$~\cite{Lu2016-ma}---to improve the performance of ranking nodes by importance~\cite{Sheikhahmadi2022multi,LIU2015importance,yang2018multi,gong2016virtual,cao2018novel}. However, the same topological attributes are used to improve ranking irrespective of the specific network’s structure. By contrast, many real applications involve node attributes that are not derivable from the graph itself, whose meaning is dictated by the problem domain.

Research on ranking nodes with external attributes---additional node-level information that is not derivable from the graph structure---is much more sparse. In \cite{benyahia2015centrality}, authors incorporate attribute information into the centrality calculation, while \cite{sanchez2014local} uses external numerical node attributes to rank outliers, which is a different task from importance ranking; \cite{lopes2017use} incorporates node attributes to perform social network analysis for the particular application of international trade. For our problem of interests, the most relevant method is AttriRank~\cite{hsu2017unsupervised}, which extends PageRank by incorporating node attributes. The main difference is an additional assumption that nodes with similar attributes should share similar rankings.

\subsection{Gap and Motivation} 

Despite ongoing research in node removal and ranking, no general method simultaneously addresses (1) structural importance, and (2) difficulty of replacement based on external attributes. While \cite{dearruda2024assigningentitiesteamshypergraph} recently extended the team formation task~\cite{juarez2021comprehensive} to investigate the issue of removing and possibly replacing nodes having certain skills, the team formation includes unique constraints (e.g., budgets, energy specifications, and multi-member teams) different from our setting.  In \cite{Romanini2021-privacy} authors explored the privacy and uniqueness of neighborhoods in social networks, but do not address structural importance. Therefore, identifying nodes that are both structurally important and difficult to replace due to unique attributes remains an open problem that we address in this work.

\section{UniqueRank} \label{sec:uniquerank} 
In this section, we introduce UniqueRank, a method for identifying nodes that are important structurally and have unique attributes. We first point out the shortcomings of naive approaches, then present our Markov Chain-based model and its refinement step, and conclude with a complexity analysis.

\subsection{Overview and Example}
We consider a network $G = \langle V, E, X \rangle$, where $V$ is the set of nodes, $E$ is the set of edges, and $X \in \mathbb{R}^{K \times N}$ is a matrix representing node attributes. Each node $i \in V$ is characterized by a vector of attribute values ${x_{1,i}, x_{2,i}, \ldots, x_{k,i}}$. In practical applications---such as terrorist or supply chain networks---nodes can correspond to individuals or facilities, edges represent interactions or transactions, and attribute values may encode roles, organizational memberships, or other domain-relevant features.

\begin{figure}[]
    \includegraphics[scale=.62]{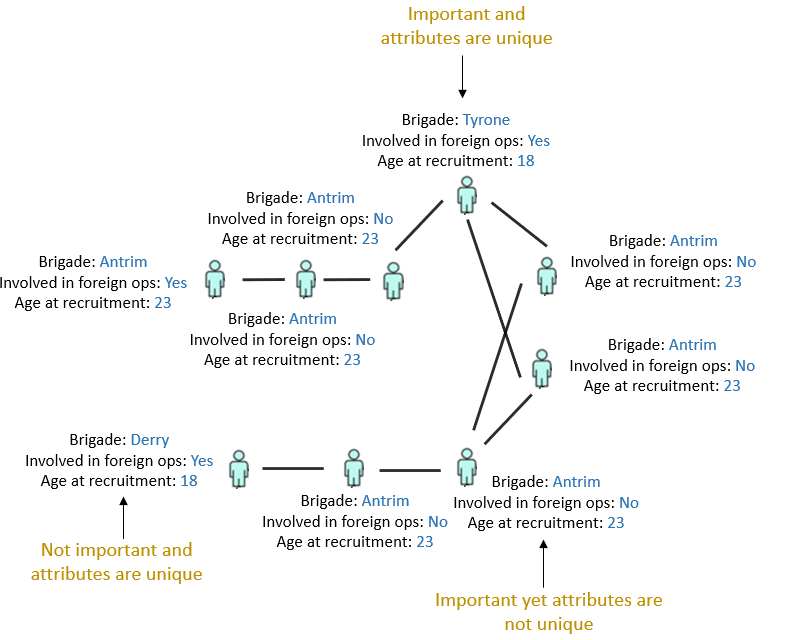}
    \caption{A (synthetic) motivating example of an important application of UniqueRank: sub-graph of a terrorist network. Attribute categories are inspired by the Provisional Irish Republican Army network.}
    \label{fig:small-ex}
\end{figure}

Figure~\ref{fig:small-ex} shows a sample subgraph of a terrorist network. The top-most node is assigned the highest rank, because it is not only structurally central with many connections but also possesses particularly unique attributes. In contrast, the third node from the left on the bottom row is well-connected but lacks unique attributes, while the left-most node on the bottom has unique attributes but few connections. Therefore, both are ranked lower than the top-most node.

\subsection{Naive approaches and issues}
\label{sec:naive}

We first consider several naive approaches for identifying structurally important nodes with unique attributes by combining measures of structural importance ($i$) and attribute uniqueness ($u$) using simple functions, such as
\[
\text{i} + \text{u}, \quad 
\text{i} \times \text{u}, \quad 
\text{or} \quad \sqrt{\text{i}^2 + \text{u}^2}.
\]
However, these combined metrics often need to be tailored to the specific topology and attribute distribution of each network, limiting their general applicability.

To visualize structural importance and uniqueness of nodes, we map node importance to the x-axis and attribute uniqueness to the y-axis in a two-dimensional plane. Ideally, nodes in the top-right corner would be both structurally important and have unique attributes in their local neighborhoods. However, we find that when plotting nodes in real networks on such a coordinate plane, no clear top-right cluster emerges as in Figure~\ref{fig:real-plot}. Although combining importance ($i$) and uniqueness ($u$) through a simple function may be effective for convex distributions, our empirical results indicate that the actual distribution is often concave. Therefore, many such functions disproportionately prioritize nodes with high attribute uniqueness but low structural importance---opposite of our goal.
This observation highlights the need for careful selection of both parameters and functional form when designing node-ranking strategies that jointly consider structural importance and uniqueness, because optimal choices may vary significantly across diverse real-world attributed networks.

\begin{figure}[H]
\centering
    \includegraphics[scale=.3]{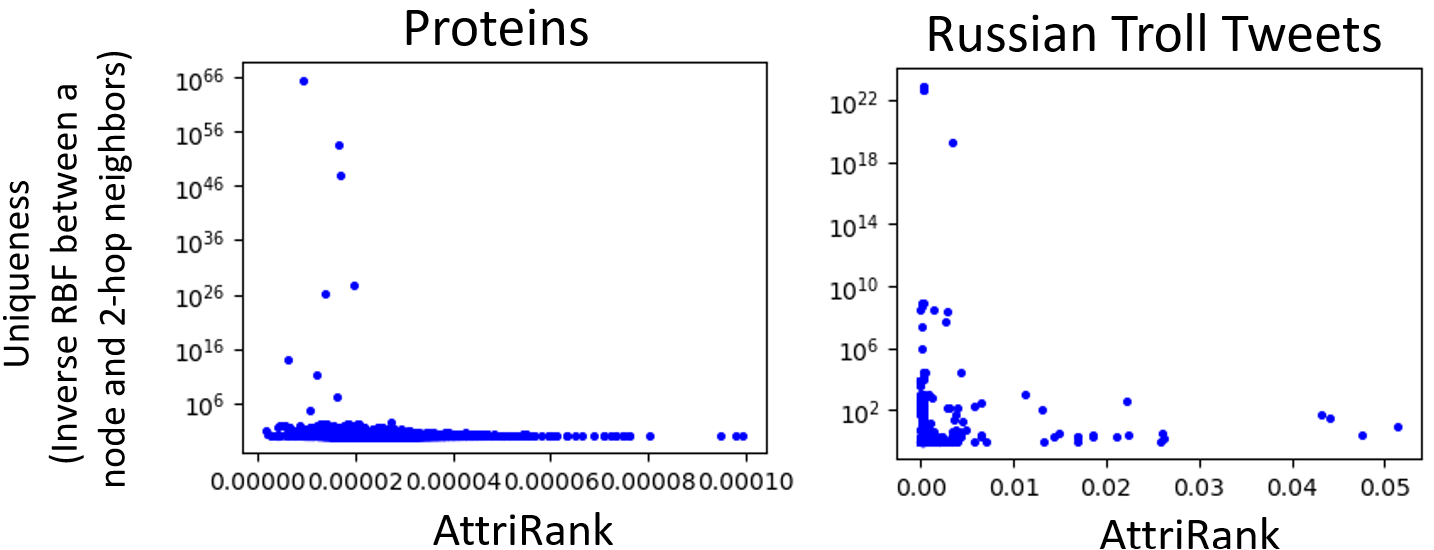}
    \caption{Example plots for real networks, where the x-axis is structural importance (AttriRank) and the y-axis is log-scale attribute uniqueness (formally defined in Equation 7, as an inverse of similarity, defined in Equation 6 and consistent with \cite{hsu2017unsupervised}, among a node and its neighbors).}
    \label{fig:real-plot}
\end{figure}

\subsection{Method} \label{sec:method}

Ranking important nodes in a network has been well studied, with one of the most widely used methods being PageRank. PageRank assumes that nodes are important if many other high-ranking nodes point to them, and is formalized as a Markov Chain with update rule: 

\begin{equation}
\pi^{(t+1)}=(1-d)\frac{1}{N}1+dP\pi^{(t)}
\end{equation}
where $\pi$ is the vector of rankings for each node at a time step $t$, $P$ is the matrix of transition probabilities among nodes, $N$ is the number of nodes in the network, and $d$ is a probability of following a link to a neighbor rather than jumping to a random node. Related literature empirically finds $d$ to be $0.85$~\cite{Page1998PageRank}.

While PageRank assesses node importance solely through structural information, real-world networks often contain node attributes, such as organizational roles or categorical labels, that can be critical for meaningful ranking. To account for this, AttriRank~\cite{hsu2017unsupervised} extends the PageRank model by incorporating node attribute similarity into the ranking process. In AttriRank, the random walk is allowed to move not only along the edges of the original graph, but also across nodes that share similar attribute values. This results in an updated ranking that balances both structural connectivity and attribute similarity. However, AttriRank emphasizes similarity rather than uniqueness, and therefore does not explicitly prioritize structurally central nodes with unique attributes.

Building on both PageRank and AttriRank, UniqueRank introduces a provably efficient method that identifies nodes with both high strucural importance and attribute uniqueness, following  two steps:
\vspace{0.5em}

\noindent(1) A \textit{Markov Chain} approach, which identifies a promising subset of nodes that are both structurally important and have unique attributes.

\vspace{0.3em}
\noindent(2) A \textit{refinement} step, which uses this subset to construct a high-confidence set of structurally important nodes with unique attributes.
\subsubsection{Markov Chain model} \label{sec:markov-chain} 
\noindent \textbf{Assumptions.} UniqueRank is grounded in the same assumptions as AttriRank, with a strong emphasis on the PageRank assumption: 

\vspace{0.5em}

\noindent\textit{PageRank assumption}. A node has a higher ranking if many high ranking nodes direct to it. 

\vspace{0.5em}

\noindent compared to the attribute assumption:

\vspace{0.5em}
\noindent\textit{Attribute assumption}. If a pair of nodes $i,j$ have similar attribute values $x_i \approx x_j$, then they should have similar ranking scores $\pi_i \approx \pi_j$.

\vspace{0.5em}

\noindent UniqueRank places greater weight on the PageRank assumption, controlled by a hyperparameter $d$.\\

\noindent Because structurally important nodes with unique attributes rank higher, the PageRank assumption specifies that such nodes are linked by other high-ranking nodes---whether due to their structural importance, unique attributes, or both. The attribute assumption, which is common in machine learning, indicates that nodes with similar attributes should share similar rankings. Although these two assumptions may sometimes contradict each other, UniqueRank prioritizes the PageRank assumption, setting $d$ to $0.85$.

\vspace{1em}
\noindent\textbf{Update rule.} Formally, the update rule for the Markov Chain model is 
\begin{equation}  \label{eq:update}
\pi^{(t+1)}=(1-d)Q\pi^{(t)} +dP\pi^{(t)}
\end{equation}
where $Q$ is a transition matrix for $H$ such that
\begin{equation}  \label{eq:transH}
q_{ij}=\frac{s_{ij}}{\Sigma_{k\in V} s_{kj}}
\end{equation}
where $P$ is a transition matrix for $G$ such that
\begin{equation}  \label{eq:transG}
p_{ij}=\frac{w_{ij}}{\Sigma_{n\in neighbors(i)} w_{in}}
\end{equation} and
\begin{equation}  \label{eq:weight}
w_{ij}=\frac{1}{\alpha + (1-\alpha) \min_{n\in neighbors(t)} s_{nt}}
\end{equation}
The vector $\pi^{(t)}$ stores the ranking of each node at time step $t$, with $\pi^{(0)}$ initialized randomly. The similarity of attributes between nodes $i$ and $j$ is denoted by $s_{ij}$. To control the trade-off between structural importance and attribute uniqueness, we introduce a hyperparameter $\alpha$: as $\alpha \rightarrow 1$, the edge weights $w_{ij}$ approach $1$, effectively recovering AttriRank rankings based primarily on attribute similarity. On the other hand, as $\alpha \rightarrow 0$, the contribution of attribute uniqueness increases. Conceptually, the model can be interpreted as a random walk:
\vspace{0.4em}

\noindent \textit{The walk on graph $G$}, with transition probabilities indicated by transition matrix $P$, aligns with the PageRank assumption and is weighted by $d$.  

\vspace{0.4em}

\noindent \textit{The walk on graph $H$}, with transition probabilities indicated by transition matrix $Q$, aligns with the AttriRank assumption and is weighted by $1-d$.

\begin{figure}[H]
    \includegraphics[scale=.52]{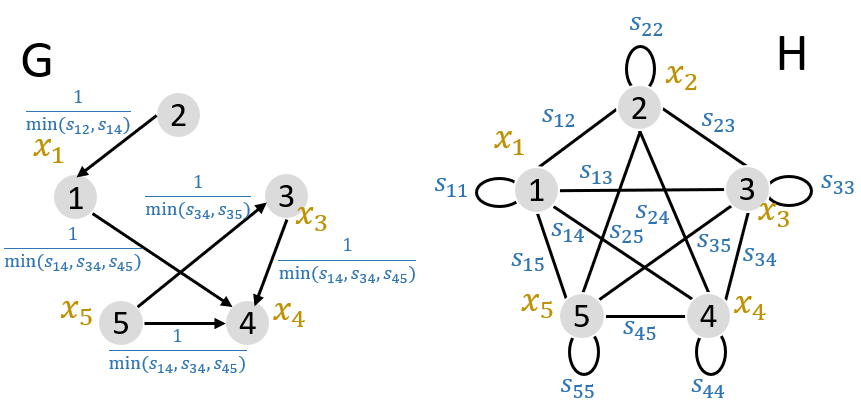}
    \caption{A conceptual view of UniqueRank’s Markov Chain, integrating two walks: one that captures the PageRank assumption (graph $G$) and one that captures the attribute assumption (graph $H$).}
    \label{fig:walks}
\end{figure}

\noindent\textbf{Random walk perspective.} From the perspective of a random walk, a single ``walker'' transitions between nodes according to two underlying graphs: $G$, representing the structural topology, and $H$, encoding attribute similarities, as in Figure~\ref{fig:walks}. A step on $G$ enforces the PageRank assumption (structural importance), while a step along $H$ enforces the attribute assumption (attribute similarity or uniqueness). The update rule combines these two random walks through a linear combination, with the damping factor $d$ controlling the probability of following a structural rather than an attribute-based transition. Consistent with PageRank, UniqueRank sets $d = 0.85$ by default, therefore prioritizing the PageRank assumption over the attribute assumption.

The topology of graph $G$ is identical to the original graph as in Figure~\ref{fig:walks}, ensuring that random walks are consistent with the true edge structure. Unlike traditional methods~\cite{Page1998PageRank,brin1998anatomy}, which assign equal transition probabilities to all neighbors, UniqueRank increases the likelihood of transitioning to neighbors with more distinct attributes. Therefore, movement from node $i$ to a neighbor $j$ is proportional to how different $j$'s attributes are, relative to other neighbors.

The transition probability matrix $P \in \mathbb{R}^{N \times N}$ encodes the chance of moving from node $i$ to node $j$ along edges of $G$, with $p_{ij}$ representing the probability of such a move. Each edge $(i, j)$ receives a weight $w_{ij}$, defined as the maximum attribute difference between node $j$ and its neighbors (see Equation~\ref{eq:weight}), therefore capturing a local neighborhood of at most two hops from $i$. To obtain $p_{ij}$, each $w_{ij}$ is normalized by the sum of all outgoing edge weights from node $i$, as formalized in Equation~\ref{eq:transG}.

Graph $H$ is the same as in AttriRank, where the topology of $H$ is a complete graph as in Figure~\ref{fig:walks}. Transition probabilities are proportional to attribute similarity $s_{ij}$, which is computed using the Radial Basis Function (RBF) kernel, consistent with AttriRank:
\begin{equation}
s_{ij} \label{eq:sim} \equiv e^{-\gamma ||x_i-x_j||^2}.
\end{equation}
Hsu et al.~\cite{hsu2017unsupervised} have shown that RBF-based similarity outperforms measures such as cosine similarity and Euclidean distance in this context.

The transition probabilities for $H$ are captured by the matrix $Q \in \mathbb{R}^{N \times N}$, where $q_{ij}$ denotes the probability of moving from node $i$ to node $j$. This is computed by normalizing $s_{ij}$ by the sum of attribute similarities from node $i$ to all other nodes, as in Equation~\ref{eq:transH}.

\vspace{.2em}
\noindent We prove that rankings of $\pi$ converge after infinitely many iterations:
\begin{proposition}  \label{th:stationary}
The vector $\pi$ always reaches a stationary probability distribution.
\end{proposition}

\begin{proof}
We denote the Markov Chain model as $R\equiv (1-d)Q+dP$ and prove that it is both irreducible and aperiodic.

\vspace{.5em}
\noindent \textit{Irreducible}: Every node can be reached from every other node. Since $q_{ij} \propto s_{ij}$, and $s_{ij}$ (defined in Equation~\ref{eq:sim}) is strictly positive, it follows that $r_{ij}>0$ for all $i, j$.  

\vspace{.5em}
\noindent \textit{Aperiodic}: Each node has a nonzero probability of returning to itself because $q_{ij} \propto s_{ij}$, and $s_{ij}>0$. Thus, $r_{ij}>0$.
% \end{itemize}
Any irreducible and aperiodic Markov Chain model converges to a stationary probability distribution after infinitely many iterations~\cite{Häggström_2002irreducible}. 
\end{proof}

\noindent\subsubsection{Refinement step.}To solidify the final selection of nodes that are both structurally important and have distinct attributes with respect to neighbors, we apply a refinement procedure using the top-ranked nodes outputted by the Markov Chain. This procedure systematically identifies the nodes that most often surpass the initial top-ranked candidate nodes in both structural importance and attribute uniqueness. 

\vspace{0.3em}
\noindent\textbf{Algorithm description.} The refinement algorithm considers the two-dimensional space where each node is represented by a point with coordinates given by its structural importance (x-axis) and attribute uniqueness (y-axis), as shown in Figure~\ref{fig:boxes}. 

Given the set $T$ of $k$ top-ranked nodes identified by the Markov Chain step, the algorithm considers, for each $j\in T$, the region containing all nodes that have both \emph{greater importance} and \emph{greater uniqueness} than node $j$--- that is, all points strictly to the right and above $j$ in this 2D plane. Therefore, for each node $i$, the algorithm counts how many such regions it occupies; that is, for how many $j\in T$ it holds that $a_i > a_j$ and $u_i > u_j$. This count is denoted as $b(i)$. 

Finally, the $k$ nodes with the largest $b(i)$ values are selected as the refined top-$k$ set. In case of ties, nodes with larger combined score ($a_i+u_i$) are selected. This ensures that the final set contains the nodes that are most frequently dominant in both dimensions with respect to the initial top candidates, and that no node outside the set is both more important and more unique than a node within the set. 

\begin{algorithm}{\footnotesize
\DontPrintSemicolon  % optional: removes semicolons  
\caption{\footnotesize Algorithm to construct a higher confidence set of structurally important nodes with unique attributes}\label{alg:box}
\SetKwInOut{Input}{Input}
\SetKwInOut{Output}{Output}
\Input{$S = \{(a_i, u_i) \mid i \in V$, $a_i$ is the importance of $ i$, $u_i $ is the uniqueness of $ i\}$, and indices $T$ corresponding to the top $k$ nodes from the Markov Chain model step}
\Output{Indices of the final top $k$ nodes}
min\_uniqueness $\gets 1$\;
min\_importance $\gets 1$\;
\For{$i \in T $}{
$\textrm{min\_importance} \gets \min(\textrm{min\_importance}, a_i)$\;
$\textrm{min\_uniqueness}\gets \min(\textrm{min\_uniqueness}, u_i)$\;}
$b \gets \{\}$\;
\For{$i\in V$}{
    \lIf{$a_i < \textrm{min\_importance}$ or $u_i <\textrm{min\_uniqueness}$} {
    continue }
    \For{$j \in T $}{
        \lIf{$a_i > a_j$ and $u_i > u_j$} {
        $b[i] \gets b[i]+ 1$ }}}
Sort keys in $b$ by value\;
Return the top $k$ keys in $b$ by value; given a tie, return keys of the tied value that have greater $a_i + u_i$ until $k$ keys have been returned\;}

\end{algorithm}

\vspace{0.3em}
\noindent Let ``importance" be measured by the AttriRank metric. The ``uniqueness" $u_i$  of node $i$ is defined as: 
\begin{equation} \label{eq:uniq}
u_i=\frac{1}{\frac{1}{|N(i)|}\sum_{j\in N(i)}s_{ij}}
\end{equation}
where $N(i)$ is the set of neighbors of $i$, and $s_{ij}$ is the similarity between nodes $i$ and $j$.

\begin{proposition}
\label{th:box}
Algorithm~\ref{alg:box}, given node importance and uniqueness scores as input, returns a set of the top-$k$ nodes such that for any node $n$ outside this set, there does not exist a node $m$ inside the set for which $n$ is both more important and more unique than $m$.
\end{proposition}

\begin{figure}[H]
\centering
    \includegraphics[scale=.38]{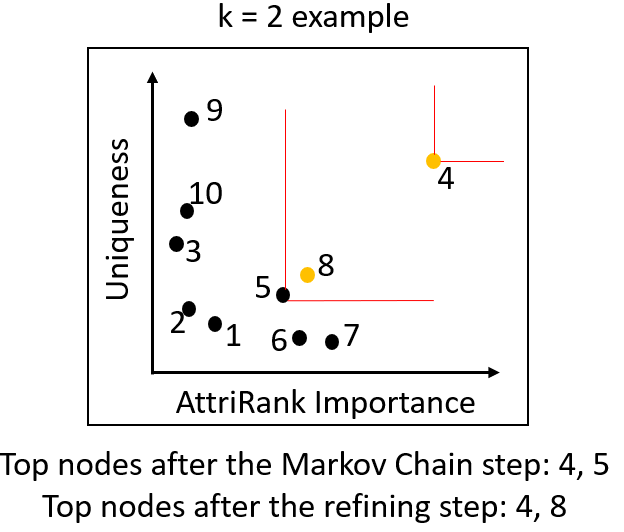}
    \caption{An example illustration of the refinement step for selecting the top 2 nodes in a 10-node graph. The sampled nodes (connected to red lines) after the Markov Chain model step are used to identify a higher confidence set of nodes (yellow) during the refinement step.}
    \label{fig:boxes}
\end{figure}

\begin{proof}

Assume that Proposition~\ref{th:box} is false. In particular, assume node $m \in S$ is among the top $k$ structurally important and unique nodes, but there is a node $n \notin S$ that is both more important and more unique yet not included in the top $k$. By Algorithm~\ref{alg:box}, $b(m) \ge b(n)$. However, if $n$ exceeds $m$ in both importance and uniqueness, it follows that $b(n) \ge b(m)$. The only way to satisfy both inequalities is if $b(m)=b(n)$. During this tie, Algorithm~\ref{alg:box} selects nodes with higher $u_i$ (uniqueness) first, which holds for $n$. Therefore, $n$ must also be in the top $k$, contradicting the initial assumption that $n\notin S$.

\end{proof}

\subsection{Computational Complexity} \label{sec:computational-complexity}

The Markov Chain step in UniqueRank is provably efficient: \cite{hsu2017unsupervised} show that AttriRank achieves linear time complexity in the number of edges by using an iterative approximation to compute ranking scores, rather than solving for the exact stationary distribution directly. UniqueRank inherits this approximation-based efficiency and adds a one-time cost of $O(|V| + |E|)$ for computing edge weights $w_{ij}$. The subsequent refinement step in Algorithm~\ref{alg:box} has an average complexity of $O(k|V| + k \log k)$, where $k$ is the number of top-ranked nodes (i.e., structurally important with unique attributes) produced by the Markov Chain model, which is often set to be very small. In the unlikely event that the output of the Markov Chain model contradicted the PageRank assumption---that is, if high ranked nodes did not have edges with other high ranking nodes---the worst case complexity for Algorithm~\ref{alg:box} could reach $O(k|V|+|V|\log|V|)$. However, the PageRank assumption is widely used and well validated (e.g., in Google), and for simple, symmetric networks, we verify that the Markov Chain model's results align with this assumption in practice.

\subsection{Simple Symmetric Networks} \label{sec:symmetric} 

We first show that the Markov Chain model in UniqueRank indeed ranks nodes with unique attributes and high structural importance more highly on symmetric graphs that have simple attribute assignments. A graph is symmetric if, for any two pairs of adjacent vertices $(u_1,v_1)$ and $(u_2,v_2)$, there exists an automorphism $f:V(G) \to V(G)$ such that $f(u_1)=u_2$ and $f(v_1) = v_2$. In such graphs, many nodes share identical structural roles. To distinguish among these, we introduce slight attribute modifications—labelled as ‘perturbed’ attributes—to a subset of nodes, while the remaining nodes retain the default attribute values. In this setting, our ground truth is that nodes with perturbed attributes should be considered more unique than their structurally identical peers.

\begin{observation}
If the damping factor $d = 1$, the Markov chain component of UniqueRank assigns higher ranks to nodes with unique attributes compared to structurally identical nodes with default attributes.
\end{observation}

\noindent To support this observation, consider a cycle graph, where each node has exactly two neighbors; under the graph’s symmetry, all nodes are structurally equivalent. If all nodes share the same attributes except for one node possessing unique attributes, the transition matrix $P$ in Equation~\ref{eq:transG}---which incorporates attribute similarity---assigns uniform transition probabilities among nodes with default attributes, but increases the probability of transitioning {\em to} the unique node from its neighbors, since its attribute similarity with others is lower, causing its similarity-based transition weights to be higher. Therefore, the Markov chain random walk is more likely to visit the unique node, and thus the stationary distribution assigns it a higher rank compared to the structurally equivalent nodes with default attributes.

This argument extends to other symmetric network structures as well. We validated this behavior empirically by applying UniqueRank to a range of symmetric graphs with the described attribute configurations. In each instance, UniqueRank assigned higher ranks to those nodes that were structurally important and possessed unique attributes, in agreement with the ground truth.

\begin{table*}[t]{
\begin{tabular}{cccccccc}
\toprule
\textbf{Network} &
  \textbf{\makecell{Edge\\Type}} &
  \textbf{\makecell{\# \\Nodes}} &
  \textbf{\makecell{Node}} &
  \textbf{\makecell{\# \\Edges}} &
  \textbf{\makecell{Edge \\Meaning}} &
  \textbf{\makecell{\# \\Attr}} &
  \textbf{Attribute types} \\ \midrule
\makecell{Provisional Irish \\Republican Army} &
  Undir &
  391 &
  Person &
  864 &
  \makecell{Involvement in an \\act together, friends, \\blood relatives, \\married, etc} &
  22 &
  {\scriptsize\makecell{ Gender;   university; marital status; age at \\recruitment; brigade: Antrim, Derry,   \\Armagh, Down, Tyrone, Fermanagh; senior \\leader; gunman; involvement in:   violent \\activity, nonviolent activity, foreign ops, \\bank   robbery/kidnapping/hijacking/drugs, \\violent activity in foreign ops;   improvised \\explosive device(ied) constructor; \\improvised explosive   device(ied) planter}} \\ \midrule
\makecell{Russian Troll Tweets  \cite{linvill2020troll}} &
  Dir &
  1211 &
  \makecell{User   \\handle} &
  2855 &
  \makecell{Interaction   between \\user handles} &
  20 &
  {\scriptsize\makecell{ Account   category: right troll, left troll, \\fearmonger, hashtag gamer, NonEnglish;   \\account type: Russian, right, left, Ukranian, \\hashtager, Koch; activity; max   followers; \\max following; min followers; min following}} \\ \midrule
Facebook-107 \cite{mcauley2012learning,yang2020scaling} &
  Undir &
  1034 &
  \multirow{3}{*}{Person} &
  26749 &
  \multirow{3}{*}{Friend} &
  168 &
  \multirow{3}{*}{\scriptsize\makecell{ Degree;   education concentration; work \\position; employer; languages known (all   \\one-hot encoded)}} \\ 
Facebook-348 \cite{mcauley2012learning,yang2020scaling} &
  Undir &
  224 &
   &
  3192 &
   &
  41 &
   \\ 
Facebook-3437 \cite{mcauley2012learning,yang2020scaling} &
  Undir &
  534 &
   &
  4813 &
   &
  70 &
   \\ \midrule
\makecell{Supply-chain-5 {\scriptsize(food} \\{\scriptsize preparations)} \cite{williams2018real}} &
  Dir &
  27 &
  \multirow{8}{*}{\makecell{Supply   \\chain \\stage}} &
  31 &
  \multirow{8}{*}{\makecell{Transfer   of items \\between stages}} &
  8 &
  \multirow{8}{*}{{\scriptsize \makecell{Stage   cost; supply chain depth; stage time; \\stage type: manuf, part, retail, trans,   dist}}} \\ 
\makecell{Supply-chain-7 {\scriptsize (construction} \\{\scriptsize machinery and equipment)} \cite{williams2018real}} &
  Dir &
  38 &
   &
  78 &
   &
  8 &
   \\ 
\makecell{Supply-chain-15 {\scriptsize (soap and} \\{\scriptsize detergents) }\cite{williams2018real}} &
  Dir &
  133 &
   &
  164 &
   &
  8 &
   \\ 
\makecell{Supply-chain-20 {\scriptsize (computer }\\{\scriptsize peripheral equipment)} \cite{williams2018real}} &
  Dir &
  156 &
   &
  169 &
   &
  8 &
   \\ 
\makecell{Supply-chain-27 \\{\scriptsize (electromedical and }\\{\scriptsize electrotherapeutic  } \\{\scriptsize apparatus)} \cite{williams2018real}} &
  Dir &
  482 &
   &
  941 &
   &
  8 &
   \\ 
\makecell{Supply-chain-28 {\scriptsize (computer} \\{\scriptsize storage devices)} \cite{williams2018real}} &
  Dir &
  577 &
   &
  2262 &
   &
  8 &
   \\ 
\makecell{Supply-chain-32 {\scriptsize (perfumes,} \\{\scriptsize cosmetics and other toilet}   \\{\scriptsize preparations)} \cite{williams2018real}} &
  Dir &
  844 &
   &
  1685 &
   &
  8 &
   \\ 
\makecell{Supply-chain-37 {\scriptsize (industrial} \\{\scriptsize organic chemicals)} \cite{williams2018real}} &
  Dir &
  1479 &
   &
  2069 &
   &
  8 &
   \\ \bottomrule
\end{tabular}
}\caption{Real-world networks used for evaluation.}
\label{tab:datasets}
\end{table*}

\section{Experiments} \label{sec:experiments} 

We conducted a comprehensive evaluation of our approach across various real-world networks. To illustrate UniqueRank's value, we plot each node's uniqueness and structural importance in a two-dimensional coordinate plane. The horizontal axis represents importance (AttriRank) and the vertical axis reflects logarithmic attribute uniqueness (Equation~\ref{eq:uniq}). In Figure \ref{fig:vis-results}, nodes selected by UniqueRank are highlighted in red, while those by AttriRank are the right-most nodes in the plot. The red nodes selected by UniqueRank incorporate attribute uniqueness into their evaluation, allowing for a slight trade-off in structural importance. In contrast, the right-most nodes selected by AttriRank focus solely on structural importance to identify top-ranked nodes. This visualization includes networks\footnote{https://networkrepository.com/} of enzymes \cite{Borgwardt2005protein,Schomburg2004-vm}, proteins \cite{Dobson2003-zz}, organic molecules \cite{chmiela2017machine} and more.

\begin{figure}[H]
    \includegraphics[scale=.4]{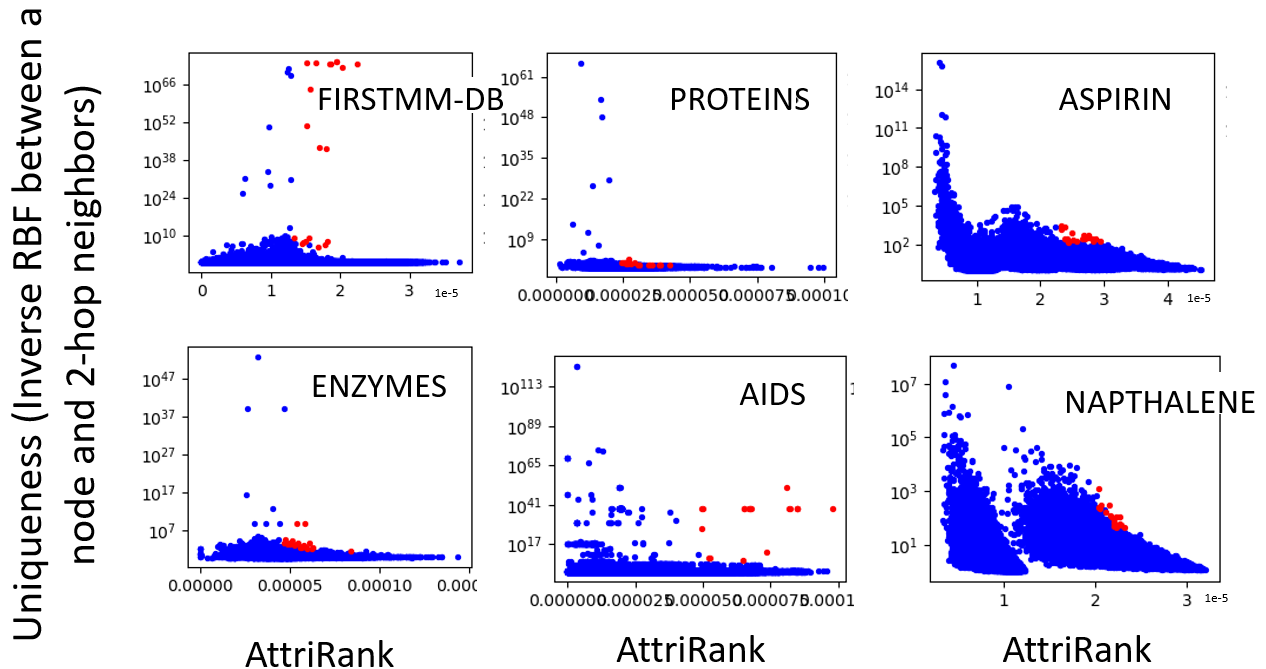}
    \caption[Uniqueness (log scale) and importance of nodes that UniqueRank selects on real networks]{Uniqueness (log scale) and importance of nodes of FIRSTMM-DB~\cite{Neumann2016-jb}, PROTEINS~\cite{Dobson2003-zz}, ASPIRIN~\cite{chmiela2017machine}, ENZYMES~\cite{Borgwardt2005protein,Schomburg2004-vm}, {AIDS\footnote{https://wiki.nci.nih.gov/display/NCIDTPdata/AIDS+Antiviral+Screen+Data}}~\cite{riesen2008iam}, NAPHTHALENE~\cite{chmiela2017machine},  visualized on a logarithmic y-axis scale. }
    \label{fig:vis-results}
\end{figure}

\noindent Next, we begin discussion of the evaluation strategy.

\subsection{Evaluation metric of efficiency reduction} \label{sec:efficiency} 
In real networks such as terrorist networks, social networks, and supply chain networks, identifying structurally important nodes is critical as their removal may greatly affect efficiency of the network. We use Latora and Marchiori’s definition of network efficiency \cite{latora2001efficient}, calculated as
\begin{equation} \label{eq:efficiency} E(G)=\Sigma_{i \neq j \in G} \frac{1}{d_{ij}}\end{equation}
where $d_{ij}$ is the shortest-path distance between nodes $i$ and $j$. To quantify efficiency reduction, we compare the original graph $G$ with a modified graph $G_{\text{mod}}$, which reflects the network after removal and attempted replacement of a key node:

\begin{equation} 1-\frac{E(G_{mod})}{E(G)}\end{equation}

\noindent In real settings, especially in social networks, the network may respond to node removal by selecting a nearby replacement. Consequently, this measure captures how efficiency changes after removal and attempted replacement of a node. To capture the local effect of disruption and reconfiguration, our computation focuses on the node selected for removal and its two-hop neighborhood. By focusing on this subset, our measurement captures the impact on the nodes most directly affected by the removal, without being diluted by distant connections that are less affected. This provides a clearer understanding of how local structural relationships adapt after node removal and attempted replacement, which is particularly relevant in social networks where local ties significantly shape functionality.

\subsection{Evaluation strategy on social networks} \label{sec:evaluation} 
\noindent To evaluate, we follow steps: 

\begin{enumerate}
\item \textit{Node Removal}. Consistent with previous literature \cite{bellingeri2020link,chen2013node,JAHANPOUR20133458,Boldi2013-wp,Iyer2013-fm,BELLINGERI2014174,gallos2005stability,NIE2015248,Tian2017-en}, we remove a top-ranked node from the network.  
\item \textit{Node Replacement} \cite{missaoui2013social}. We search within the removed node's 2-hop neighborhood for a candidate node whose attribute similarity, measured using a radial basis kernel $t(x_i,x_j)$, exceeds a set threshold. If such a similar node exists, we redirect all incoming and outgoing edges from the removed node to the replacement node, effectively simulating a replacement occupying the removed node's position.  

\item \textit{Efficiency reduction computation}. We measure the drop in network efficiency by comparing $E(G)$ (efficiency of the original network) to $E(G_{mod})$ (efficiency after node removal and attempted replacement).
\end{enumerate}

\noindent We vary two parameters in these experiments: the maximum distance of a replacement node from the removed node as $k$ and minimum similarity between a replacement node and the removed node as $t(x_i,x_j)$. By testing different values for $k$ and $t(x_i,x_j)$, we investigate how UniqueRank performs under differing criteria for node replacement.

\subsection{Datasets} \label{sec:datasets}

We evaluate our approach on real terrorist, social, and supply chain networks, listed in Table~\ref{tab:datasets}. Each network is connected and contains node attributes, which may be binary or real-valued; categorical attributes are one-hot encoded. The terrorist networks consist of the undirected Provisional Irish Republican Army network and the Russian Trolls Twitter Mention Network \cite{linvill2020troll}. The supply chain networks \cite{williams2018real} include supply chains for airplane parts, toys, and household items. 
We also evaluate our approach on Facebook social networks \cite{mcauley2012learning}, pre-processed by \cite{yang2020scaling}, in which user attributes are anonymized to protect privacy. These networks have been widely used in various literature \cite{meng2019coembedding,Yang2013community,zhang2018anrl}.

\subsection{Baseline methods}
We compare our method with both traditional and more recent node-importance ranking methods on real networks. The traditional methods include centrality metrics of degree, closeness, and eigenvector, which have been used for decades and are still common in social science literature. We also compare our method to AttriRank, an extension of PageRank that considers external node attributes and is widely used.

Because our method incorporates attribute \emph{uniqueness} to select top nodes, and we could not find existing approaches that do this, we introduce a self-constructed naive method motivated by our evaluation strategy in \S\ref{sec:evaluation}. In this naive method, nodes are first ranked by structural importance. Among these, we select the top $k$ nodes whose minimum distance to any other node with a similarity above a designated threshold (ranging from $0.7$ to $1.0$) exceeds two hops. Thus, for each selected node $i$, no node $j$ within a two-step neighborhood has an attribute similarity above the threshold. As the threshold varies, so does the final selection of top nodes. We evaluate this baseline for each threshold in terms of both the distance to potential replacement nodes (\S\ref{sec:distance}) and the efficiency reduction achieved upon their removal and attempted replacement (\S\ref{sec:eff-reduc}).

\subsection{Investigating distance of potential replacement nodes} \label{sec:distance}
We first analyze the distance between removed nodes, selected according to various ranking methods, and candidate replacement nodes with sufficiently similar attributes (i.e., nodes that meet or exceed a specified RBF similarity threshold). For each threshold $t(x_i,x_j)\in{0.5,0.6,0.7,0.8}$, and for the top $k\in{3,5,10}$ ranked nodes, we observe that nodes identified by UniqueRank generally exhibit a greater average distance to the nearest acceptable replacement node than those chosen by AttriRank or by traditional centrality metrics in both terrorist and social networks. These results, reported in Table~\ref{tab:avg_dist}, show that UniqueRank tends to identify critical nodes that are not only important structurally but are also harder to substitute based on attributes. We exclude supply chain networks from this analysis because their maximum path distance is significantly smaller than that of the other networks.

\vspace{0.3em}

\noindent \textbf{Investigation of distance for brute force baseline method.} Tables~\ref{tab:avg_dist-baselinepira} and \ref{tab:avg_dist-baselinert} present the results for the naive baseline method, which incorporates both node attribute uniqueness and structural importance, evaluated across a range of similarity thresholds. For comparison, we include the corresponding distances obtained by UniqueRank for multiple datasets; additional results are provided in Appendix~\ref{app:dist}. Overall, for the PIRA terrorist network, nodes selected by UniqueRank consistently exhibit a greater average distance to their nearest sufficiently similar replacement node than those selected by the baseline method, across all similarity thresholds from $0.7$ to $1.0$. However, this trend is less consistent in the Russian Troll Tweets network, where the results are more mixed.
\begin{table}[]
\begin{tabular}{cllllll}
\toprule
\multicolumn{1}{l}{}                 &               & PIRA         & \makecell{Russian \\Troll \\Tweets} & \makecell{fb-\\107}      & \makecell{fb-\\348}       & \makecell{fb-\\3437}     \\ \midrule
\multirow{5}{*}{\makecell{threshold \\0.7, top 5}}  & \rule{0pt}{3ex} degree        & 3.6          & 4.8                  & 4.6         & 4.6          & 2.8         \\  
                                       & \rule{0pt}{3ex} closeness     & 4            & 3                    & 4.8         & 2.8          & 1           \\ 
                                       & \rule{0pt}{3ex} eigenvector   & 1.2          & 1                    & 6.4         & 2.8          & 2.8         \\ 
                                       & \rule{0pt}{3ex} AttriRank     & 2            & 4.6                  & 4.6         & 2.8          & 2.8         \\ 
                                       & \rule{0pt}{3ex} UniqueRank & \textbf{6.6} & \textbf{6.4}         & \textbf{10} & \textbf{10}  & \textbf{10} \\ \midrule
\multirow{5}{*}{\makecell{threshold \\0.5, top 5}}  & \rule{0pt}{3ex} degree        & 1.8          & 4.6                  & 4.6         & 3            & 2.8         \\ 
                                       & \rule{0pt}{3ex} closeness     & 1.6          & 3                    & 3.4         & 2.8          & 1           \\ 
                                       & \rule{0pt}{3ex} eigenvector   & 1            & 1                    & 6.4         & 2.8          & 2.8         \\ 
                                       & \rule{0pt}{3ex} AttriRank     & 1            & 4.6                  & 4.6         & 1.2          & 2.8         \\ 
                                       & \rule{0pt}{3ex} UniqueRank & \textbf{4.2} & \textbf{6.4}         & \textbf{10} & \textbf{10}  & \textbf{10} \\ \midrule
\multirow{5}{*}{\makecell{threshold \\0.7, top 10}} & \rule{0pt}{3ex} degree        & 3.5          & 4.7                  & 4.7         & 4.6          & 1.9         \\ 
                                       & \rule{0pt}{3ex} closeness     & 3.9          & 3.9                  & 5.6         & 3.7          & 1.9         \\ 
                                       & \rule{0pt}{3ex} eigenvector   & 1.5          & 1                    & 5.5         & 2.8          & 1.9         \\ 
                                       & \rule{0pt}{3ex} AttriRank     & 1.7          & \textbf{4.8}                  & 4.7         & 3.7          & 2.8         \\  
                                       & \rule{0pt}{3ex} UniqueRank & \textbf{5.2} & 4.7         & \textbf{10} & \textbf{9.1} & \textbf{10} \\ \midrule
\multirow{5}{*}{\makecell{threshold \\0.5, top 10}} & \rule{0pt}{3ex} degree        & 2.3          & 3.7                  & 3.9         & 2.9          & 1.9         \\ 
                                       & \rule{0pt}{3ex} closeness     & 2.3          & 2.9                  & 4.9         & 2.9          & 1.2         \\ 
                                       & \rule{0pt}{3ex} eigenvector   & 1.1          & 1                    & 4.8         & 2            & 1.9         \\ 
                                       & \rule{0pt}{3ex} AttriRank     & 1.2          & 3.8                  & 4.6         & 2            & 2.1         \\ 
                                       & \rule{0pt}{3ex} UniqueRank & \textbf{2.7} & \textbf{4.7}         & \textbf{10} & \textbf{9.1} & \textbf{10} \\ \hline
\end{tabular}
\caption{Average distance between a removed node and a similar node passing a similarity threshold of $\{0.5,0.7\}$ over the top $\{5,10\}$ nodes selected by UniqueRank, AttriRank, and traditional centrality methods.}
\label{tab:avg_dist}
\end{table}

\begin{table}
\resizebox{\columnwidth}{!}{%
\begin{tabular}{cc|c|cccccc}
\toprule \multirow[t]{2}{*}{} & &  & \multicolumn{6}{|c}{Threshold for baseline} \\
\hline & & \makecell{Unique\\-Rank} & 0.95 & 0.9 & 0.85 & 0.8 & 0.75 & 0.7 \\
\midrule \multirow{6}{*}{\makecell{Threshold for \\efficiency \\reduction}} & 0.95 & 8 & 6.2 & 6.2 & 6.2 & 6.6 & 6.2 & 6.2 \\
 & 0.9 & 8 & 6.2 & 6.2 & 6.2 & 6.6 & 6.2 & 6.2 \\
 & 0.85 & 8 & 5.8 & 5.8 & 5.8 & 6.2 & 5.8 & 5.8 \\
 & 0.8 & 8 & 4.2 & 4.2 & 4.2 & 5 & 5 & 5 \\
& 0.75 & 6.6 & 4.2 & 4.2 & 4.2 & 4.8 & 5 & 5 \\
 & 0.7 & 6.6 & 4.2 & 4.2 & 4.2 & 4.8 & 5 & 5 \\
\bottomrule
\end{tabular}}
\caption{For PIRA, the average distance between a removed node and a similar node passing a similarity threshold in the range $\{0.7,...,1\}$ over the top $5$ nodes selected by UniqueRank and the brute force baseline method, across thresholds in the range $\{0.7,...,1\}$.}
\label{tab:avg_dist-baselinepira}
\end{table}

\begin{table}
\resizebox{\columnwidth}{!}{%
\begin{tabular}{cc|c|cccccc}
\toprule \multirow[t]{2}{*}{} & &  & \multicolumn{6}{|c}{Threshold for baseline} \\
\hline & & \makecell{Unique\\-Rank} & 0.95 & 0.9 & 0.85 & 0.8 & 0.75 & 0.7 \\
\midrule \multirow{6}{*}{\makecell{Threshold for \\efficiency \\reduction}} & 0.95 & 10 & 10 & 10 & 10 & 10 & 10 & 10 \\
& 0.9 & 6.6 & 6.6 & 10 & 10 & 10 & 10 & 10 \\
& 0.85 & 6.6 & 6.6 & 8.4 & 10 & 10 & 10 & 10 \\
& 0.8 & 6.4 & 4.6 & 6.6 & 8.2 & 10 & 10 & 10 \\
& 0.75 & 6.4 & 4.6 & 6.6 & 8.2 & 10 & 10 & 10 \\
& 0.7 & 6.4 & 4.6 & 5 & 6.6 & 8.4 & 8.4 & 10 \\
\bottomrule
\end{tabular}}
\caption{For Russian Troll Tweets, the average distance between a removed node and a similar node passing a similarity threshold in the range $\{0.7,...,1\}$ over the top $5$ nodes selected by UniqueRank and the brute force baseline method, across thresholds in the range $\{0.7,...,1\}$.}
\label{tab:avg_dist-baselinert}
\end{table}

\subsection{Investigating efficiency reduction} \label{sec:eff-reduc} 
We evaluate efficiency reduction for removal and attemped replacement of nodes identified by UniqueRank, AttriRank, and traditional centrality-based metrics over multiple parameter settings: $k\in{3,5,10}$ (number of top-ranked nodes) and similarity thresholds $t(x_i, x_j)\in{0.5,0.6,0.7,0.8}$. These comparisons are carried out across terrorist networks, supply chain networks, and Facebook social network components.

\begin{table*}[t]
\resizebox{\textwidth}{!}{%
\begin{tabular}{ccccccccccccccc}
\hline
\multicolumn{1}{l}{} &
   &
  \rule{0pt}{4ex} PIRA &
  \makecell{Russian \\Troll \\Tweets} &
  \makecell{fb\\-107} &
  \makecell{fb\\-348} &
  \makecell{fb\\-3437} &
  \makecell{sc-5} &
  \makecell{sc-7} &
  \makecell{sc-15} &
  \makecell{sc-20} &
  \makecell{sc-27} &
  \makecell{sc-28} &
  \makecell{sc-32} &
  \makecell{sc-37} \\ \midrule
\multirow{5}{*}{\makecell{{\scriptsize threshold} \\0.7, {\scriptsize top} 5}} &
  \rule{0pt}{3ex} {\scriptsize degree} &
  0.1997 &
  0.1546 &
  0.0044 &
  0.0127 &
  0.0069 &
  0.1273 &
  0.1325 &
  0.0225 &
  0.1955 &
  \textbf{0.104} &
  0.0011 &
  0.0683 &
  0.013 \\ 
 &
  \rule{0pt}{3ex} {\scriptsize closeness} &
  0.1380 &
  0.1008 &
  0.0030 &
  0.0100 &
  0.0043 &
  0.1273 &
  0.1325 &
  0.0234 &
  0.2791 &
  \textbf{0.104} &
  0.0011 &
  0.0463 &
  0.0414 \\ 
 &
  \rule{0pt}{3ex} {\scriptsize eigenvector} &
  0.0406 &
  0.0519 &
  0.0057 &
  0.0113 &
  0.0069 &
  0 &
  0 &
  0 &
  0 &
  0 &
  0 &
  0 &
  0 \\ 
 &
  \rule{0pt}{3ex} {\scriptsize AttriRank} &
  0.0687 &
  0.1588 &
  0.0043 &
  0.0108 &
  0.0069 &
  \textbf{0.2013} &
  0.1325 &
  \textbf{0.6} &
  0.42 &
  0.0606 &
  0.0011 &
  0.0698 &
  0.5181 \\ 
 &
  \rule{0pt}{3ex} {\scriptsize UniqueRank} &
  \textbf{0.2638} &
  \textbf{0.2107} &
  \textbf{0.0125} &
  \textbf{0.0131} &
  \textbf{0.0163} &
  \textbf{0.2013} &
  \textbf{0.2301} &
  \textbf{0.6} &
  \textbf{0.4536} &
  0.103 &
  \textbf{0.0506} &
  \textbf{0.0819} &
  \textbf{0.5387} \\ \midrule
\multirow{5}{*}{\makecell{{\scriptsize threshold} \\0.5, {\scriptsize top} 5}} &
  \rule{0pt}{3ex} {\scriptsize degree} &
  0.1002 &
  0.1403 &
  0.0044 &
  0.0107 &
  0.0069 &
  0.1273 &
  0.1325 &
  0.0225 &
  0.1955 &
  \textbf{0.104} &
  0.0011 &
  0.0683 &
  0.013 \\ 
 &
  \rule{0pt}{3ex} {\scriptsize closeness} &
  0.0789 &
  0.1008 &
  0.0030 &
  0.0100 &
  0.0043 &
  0.1273 &
  0.1325 &
  0.0234 &
  0.2791 &
  \textbf{0.104} &
  0.0011 &
  0.0463 &
  0.0414 \\ 
 &
  \rule{0pt}{3ex} {\scriptsize eigenvector} &
  0.0420 &
  0.0519 &
  0.0057 &
  0.0113 &
  0.0069 &
  0 &
  0 &
  0 &
  0 &
  0 &
  0 &
  0 &
  0 \\ 
 &
  \rule{0pt}{3ex} {\scriptsize AttriRank} &
  0.0691 &
  0.1588 &
  0.0043 &
  0.0088 &
  0.0069 &
  \textbf{0.2013} &
  0.1325 &
  \textbf{0.6} &
  0.42 &
  0.0613 &
  0.0011 &
  0.0698 &
  0.5181 \\ 
 &
  \rule{0pt}{3ex} {\scriptsize UniqueRank} &
  \textbf{0.2414} &
  \textbf{0.2107} &
  \textbf{0.0125} &
  \textbf{0.0131} &
  \textbf{0.0163} &
  \textbf{0.2013} &
  \textbf{0.2301} &
  \textbf{0.6} &
  \textbf{0.4536} &
  0.1037 &
  \textbf{0.0506} &
  \textbf{0.0819} &
  \textbf{0.5387} \\ \midrule
\multirow{5}{*}{\makecell{{\scriptsize threshold} \\0.7, {\scriptsize top} 10}} &
  \rule{0pt}{3ex} {\scriptsize degree} &
  0.1349 &
  \textbf{0.1383} &
  0.0048 &
  0.0112 &
  0.0081 &
  0.2027 &
  0.1068 &
  0.0255 &
  0.1125 &
  0.0445 &
  0.0141 &
  0.0693 &
  0.015 \\ 
 &
  \rule{0pt}{3ex} {\scriptsize closeness} &
  0.1316 &
  0.0972 &
  0.0047 &
  0.0104 &
  0.0082 &
  0.2136 &
  0.2299 &
  0.0262 &
  0.3031 &
  0.0577 &
  0.0352 &
  0.0468 &
  0.0377 \\ 
 &
  \rule{0pt}{3ex} {\scriptsize eigenvector} &
  0.0418 &
  0.0511 &
  0.0057 &
  0.0109 &
  0.0076 &
  0 &
  0 &
  0 &
  0 &
  0 &
  0 &
  0 &
  0 \\ 
 &
  \rule{0pt}{3ex} {\scriptsize AttriRank} &
  0.0932 &
  0.1209 &
  0.0045 &
  0.0271 &
  0.0098 &
  \textbf{0.2327} &
  0.2538 &
  \textbf{0.6} &
  0.3031 &
  0.0925 &
  0.0532 &
  0.0672 &
  \textbf{0.5488} \\ 
 &
  \rule{0pt}{3ex} {\scriptsize UniqueRank} &
  \textbf{0.1922} &
  0.1351 &
  \textbf{0.0099} &
  \textbf{0.0979} &
  \textbf{0.0177} &
  \textbf{0.2327} &
  \textbf{0.2791} &
  \textbf{0.6} &
  \textbf{0.3097} &
  \textbf{0.1275} &
  \textbf{0.2073} &
  \textbf{0.0821} &
  \textbf{0.5488} \\ \midrule
\multirow{5}{*}{\makecell{{\scriptsize threshold} \\0.5, {\scriptsize top} 10}} &
  \rule{0pt}{3ex} {\scriptsize degree} &
  0.0820 &
  0.1105 &
  0.0049 &
  0.0100 &
  0.0081 &
  0.1877 &
  0.1008 &
  0.0255 &
  0.1125 &
  0.0441 &
  0.0141 &
  0.0331 &
  0.015 \\ 
 &
  \rule{0pt}{3ex} {\scriptsize closeness} &
  0.0911 &
  0.0951 &
  0.0047 &
  0.0094 &
  0.0082 &
  0.2136 &
  0.2299 &
  0.0262 &
  0.3031 &
  0.0581 &
  0.0352 &
  0.0468 &
  0.0377 \\ 
 &
  \rule{0pt}{3ex} {\scriptsize eigenvector} &
  0.0421 &
  0.0511 &
  0.0057 &
  0.0099 &
  0.0076 &
  0 &
  0 &
  0 &
  0 &
  0 &
  0 &
  0 &
  0 \\ 
 &
  \rule{0pt}{3ex} {\scriptsize AttriRank} &
  0.0934 &
  0.1128 &
  0.0047 &
  0.0127 &
  0.0098 &
  \textbf{0.2327} &
  0.2538 &
  \textbf{0.6} &
  0.3031 &
  0.0929 &
  0.0532 &
  0.0672 &
  \textbf{0.5488} \\ 
 &
  \rule{0pt}{3ex} {\scriptsize UniqueRank} &
  \textbf{0.1582} &
  \textbf{0.1351} &
  \textbf{0.0099} &
  \textbf{0.0979} &
  \textbf{0.0177} &
  \textbf{0.2327} &
  \textbf{0.2791} &
  \textbf{0.6} &
  \textbf{0.3097} &
  \textbf{0.1279} &
  \textbf{0.2073} &
  \textbf{0.0821} &
  \textbf{0.5488} \\ \hline
\end{tabular}}
\caption{Efficiency reduction over similarity threshold for replacement in $\{0.5,0.7\}$ and averaging over the top $\{5,10\}$ nodes selected by UniqueRank, AttriRank, or centrality metrics. The greatest efficiency reduction among the five methods in each setting is in bold.}
\label{tab:eff-reduc}
\end{table*}

As shown in Table~\ref{tab:eff-reduc}, removing and attempting to replace a top-ranked node identified by UniqueRank consistently results in a greater reduction in network efficiency compared to nodes selected by AttriRank or traditional centrality metrics. This effect is more significant in terrorist and supply chain networks. In social networks, however, the magnitude of efficiency reduction---while still present---is smaller. This may reflect the limited practical relevance of node replacement in social networks, where substituting a “friend” is both less meaningful and often ambiguous in real-world contexts.

\vspace{0.4em}
\noindent \textbf{Investigation of efficiency reduction for brute force baseline method.} Tables~\ref{tab:effreducpira} and \ref{tab:effreducrt} present the average efficiency reduction achieved by the naive baseline method---which incorporates attribute uniqueness as well as structural importance---across various similarity thresholds, along with the corresponding results for UniqueRank (additional results are in Appendix \ref{app:effreduc}). Overall, UniqueRank outperforms the naive baseline in certain datasets, such as the PIRA terrorist network, across most threshold settings. However, in other datasets, including the Russian Troll Tweets network, the relative performance is more varied. We discuss the implications of these findings further in Section~\S\ref{sec:socexptakeaways}.

\begin{table*}
\resizebox{\textwidth}{!}{%
\begin{tabular}{cc|c|cccccc}
\toprule \multirow[t]{2}{*}{} & &  & \multicolumn{6}{|c}{Threshold for baseline} \\
\hline & & \makecell{UniqueRank} & 0.95 & 0.9 & 0.85 & 0.8 & 0.75 & 0.7 \\
\midrule \multirow{6}{*}{\makecell{Threshold for \\efficiency \\reduction}} & 0.95 & $0.2638\pm 0.0497$ & $0.1764\pm 0.0325$ & $0.1764\pm 0.0325$ & $0.1764\pm 0.0325$ & $0.3533\pm 0.0839$ & $0.2748\pm 0.0842$ & $0.2748\pm 0.0842$ \\
& 0.9 & $0.2638\pm 0.0497$ & $0.1764\pm 0.0325$ & $0.1764\pm 0.0325$ & $0.1764\pm 0.0325$ & $0.3533\pm 0.0839$ & $0.2748\pm 0.0842$ & $0.2748\pm 0.0842$ \\
& 0.85 & $0.2638\pm 0.0497$ & $0.1764\pm 0.0325$ & $0.1764\pm 0.0325$ & $0.1764\pm 0.0325$ & $0.3533\pm 0.0839$ & $0.2748\pm 0.0842$ & $0.2748\pm 0.0842$ \\
& 0.8 & $0.2638\pm 0.0497$ & $0.1725\pm 0.0342$ & $0.1725\pm 0.0342$ & $0.1725\pm 0.0342$ & $0.3533\pm 0.0839$ & $0.2748\pm 0.0842$ & $0.2748\pm 0.0842$ \\
& 0.75 & $0.2638\pm 0.0497$ & $0.1725\pm 0.0342$ & $0.1725\pm 0.0342$ & $0.1725\pm 0.0342$ & $0.2608\pm 0.0917$ & $0.2748\pm 0.0842$ & $0.2748\pm 0.0842$ \\
& 0.7 & $0.2638\pm 0.0497$ & $0.1725\pm 0.0342$ & $0.1725\pm 0.0342$ & $0.1725\pm 0.0342$ & $0.2608\pm 0.0917$ & $0.2748\pm 0.0842$ & $0.2748\pm 0.0842$ \\
\bottomrule
\end{tabular}}
\caption{For PIRA, efficiency reduction over similarity threshold for replacement in the range $\{0.7,...,1\}$ and averaging over the top $5$ nodes selected by UniqueRank and the brute force baseline method, across thresholds in the range $\{0.7,...,1\}$.}
\label{tab:effreducpira}
\end{table*}

\begin{table*}
\resizebox{\textwidth}{!}{%
\begin{tabular}{cc|c|cccccc}
\toprule \multirow[t]{2}{*}{} & &  & \multicolumn{6}{|c}{Threshold for baseline} \\
\hline & & \makecell{UniqueRank} & 0.95 & 0.9 & 0.85 & 0.8 & 0.75 & 0.7 \\
\midrule \multirow{6}{*}{\makecell{Threshold for \\efficiency \\reduction}}  & 0.95 & $0.2783\pm 0.0509$ & $0.3023\pm 0.0413$ & $0.2785\pm 0.0622$ & $0.2847\pm 0.0619$ & $0.2258\pm 0.0752$ & $0.2258\pm 0.0752$ & $0.2259\pm 0.0751$ \\
& 0.9 & $0.207\pm 0.0864$ & $0.2227\pm 0.0867$ & $0.2785\pm 0.0622$ & $0.2847\pm 0.0619$ & $0.2258\pm 0.0752$ & $0.2258\pm 0.0752$ & $0.2259\pm 0.0751$ \\
& 0.85 & $0.207\pm 0.0864$ & $0.2227\pm 0.0867$ & $0.2328\pm 0.0806$ & $0.2847\pm 0.0619$ & $0.2258\pm 0.0752$ & $0.2258\pm 0.0752$ & $0.2259\pm 0.0751$ \\
& 0.8 & $0.2107\pm 0.0842$ & $0.1588\pm 0.0888$ & $0.1741\pm 0.0822$ & $0.226\pm 0.0750$ & $0.2258\pm 0.0752$ & $0.2258\pm 0.0752$ & $0.2259\pm 0.0751$ \\
& 0.75 & $0.2107\pm 0.0842$ & $0.1588\pm 0.0888$ & $0.1741\pm 0.0822$ & $0.226\pm 0.0750$ & $0.2258\pm 0.0752$ & $0.2258\pm 0.0752$ & $0.2259\pm 0.0751$ \\
& 0.7 & $0.2107\pm 0.0842$ & $0.1588\pm 0.0888$ & $0.1715\pm 0.0832$ & $0.2235\pm 0.0765$ & $0.2232\pm 0.0767$ & $0.2232\pm 0.0767$ & $0.2259\pm 0.0751$ \\
\bottomrule
\end{tabular}}
\caption{For Russian Troll Tweets, efficiency reduction over similarity threshold for replacement in the range $\{0.7,...,1\}$ and averaging over the top $5$ nodes selected by UniqueRank and the brute force baseline method, across thresholds in the range $\{0.7,...,1\}$.}
\label{tab:effreducrt}
\end{table*}

\subsection{Sensitivity to the $\alpha$ Parameter}\label{sec:exp-hyperparameters}

We investigate how varying the trade-off parameter $\alpha$ in UniqueRank affects the selection of top nodes. The parameter $\alpha$ controls the balance between structural importance and attribute uniqueness in the ranking process. As shown in Figure~\ref{fig:vis-results2} for the Provisional Irish Republican Army terrorist network and Figure~\ref{fig:vis-results3} for the Russian Troll Tweets network, decreasing the value of $\alpha$ increases the weight given to attribute uniqueness, resulting in the selection of nodes that are less structurally central but more difficult to replace based on their attributes.
\begin{figure}[H]
    \includegraphics[scale=.28]{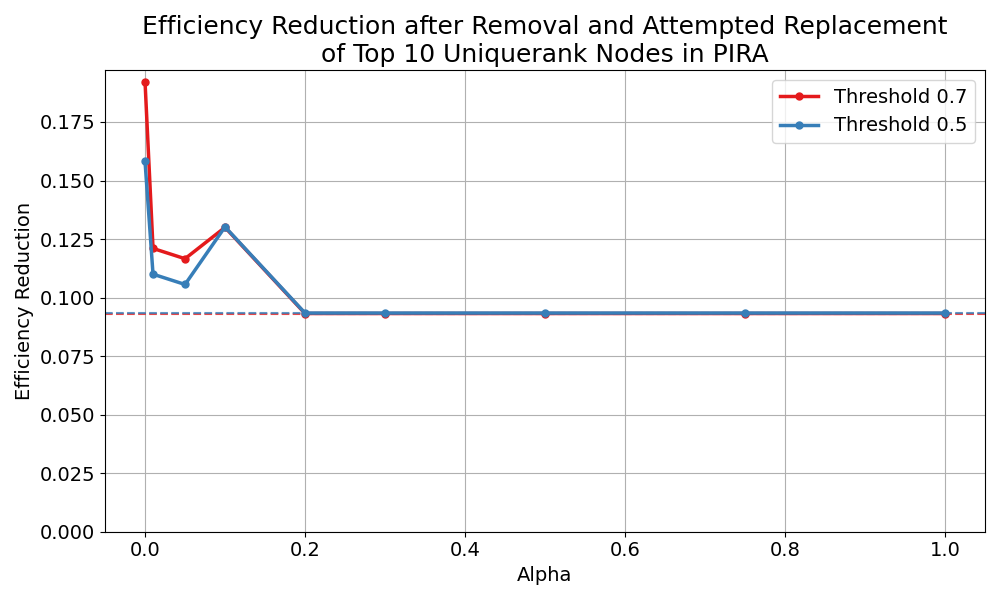}
    \caption{Efficiency reduction after removal and attempted replacement under similarity thresholds of $0.5, 0.7$ for the top 10 UniqueRank nodes in PIRA.}
    \label{fig:vis-results2}
\end{figure}
\begin{figure}[H]
    \includegraphics[scale=.28]{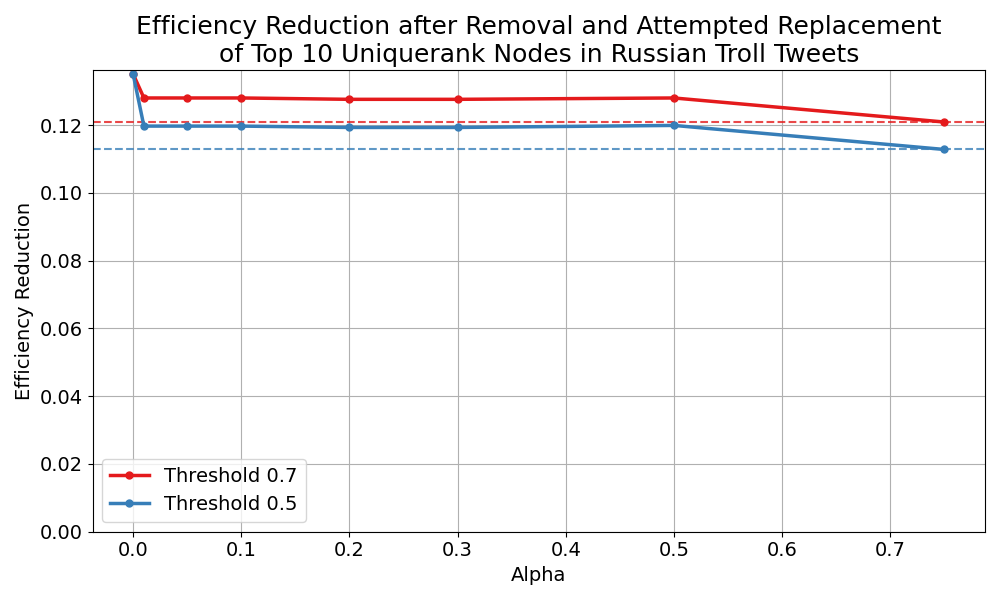}
    \caption{Efficiency reduction after removal and attempted replacement under similarity thresholds of $0.5, 0.7$ for the top 10 UniqueRank nodes in Russian Troll Tweets.}
    \label{fig:vis-results3}
\end{figure}

\subsection{Takeaway from social network experiments}
\label{sec:socexptakeaways}

UniqueRank performs well compared to both traditional centrality metrics---such as degree, closeness, and eigenvector---that have long been standard in the social sciences, and AttriRank, which is an extension of widely used PageRank that incorporates node attributes. While UniqueRank's advantage over the naive brute force method is less consistent, it is important to note that the brute force approach requires manual threshold tuning and is significantly less computationally efficient: performing $O(|V|^2)$ similarity comparisons for each threshold, and taking $O(|V|^3)$ time for efficiency reduction via all-pairs shortest path computations. Further, these steps must be repeated for every candidate node and across multiple thresholds, leading to significant computational overhead. On the other hand, UniqueRank is more efficient (see \S\ref{sec:computational-complexity}), achieves substantial raw efficiency reductions, and consistently performs well compared to standard baselines (Table~\ref{tab:eff-reduc})---all without extensive parameter tuning. Therefore, we recommend UniqueRank for this task, while noting that performance may vary by dataset characteristics.

Our evaluation strategy, which focuses on the removal and attempted replacement of individual nodes, aligns well with applications such as terrorist or criminal networks, where local disruptions in the network structure can have significant consequences. However, alternative approaches, such as assessing changes in global efficiency or analyzing the effects of removing multiple nodes, may be more suitable for other real-world scenarios. We propose UniqueRank as an initial step toward highlighting the practical importance of identifying nodes that are both structurally significant and challenging to replace based on external attributes. The random walk step---consistent with PageRank and AttriRank---and the subsequent refinement step together provide an efficient approach that consistently identifies nodes whose removal results in greater efficiency reduction than traditional importance metrics. Future work could explore alternative methods of identifying such nodes and more tailored evaluation strategies to specialized domains where node removal and network adaptation are common.

\section{Application to biomolecules} \label{sec:biomolecule} 
While UniqueRank is motivated by the observation that real-world networks may adaptively replace removed nodes with similar and nearby alternatives, its applicability extends to domains where physical node replacement is not possible, such as biomolecular structures. In these contexts, ranking node importance remains highly relevant~\cite{wang2021statistical}.  For example, Kumar et al. \cite{Kumar2023-ranking} proposed a method to identify genes in a network that are more relevant for plant-pathogen interactions, while other studies \cite{zhao2014detecting,Rajeh2022-ranking} showed that critical nodes in biological networks may help identify promising treatment targets \cite{Wang_Wang_Zheng_2022mini}.

\vspace{0.4em}
\noindent \textbf{Dataset.} MD17 \cite{chmiela2017machine} (molecular dynamics) is a dataset that consists of organic molecules, where each node represents an atom, each edge represents a covalent bond, and each node attribute stores an atom's coordinates ($x$, $y$, and $z$) and the forces acting on it in these three dimensions. At 0 Kelvins (K), the net force on each molecule is zero when summing over the three coordinate axes. These force values are used in molecular dynamics simulations that raise the temperature above 0 K, adding potential energy that is then converted to kinetic energy, causing atomic movement. Because force is the negative derivative of potential energy (e.g., $F_x = -\frac{dU}{dx}$), larger forces correspond to greater movement of the atom in its $x$, $y$, and $z$ positions.

\vspace{0.4em}

\noindent\textbf{Findings.} We applied UniqueRank to several organic molecules, including aspirin, benzene, naphthalene, and toluene. We observed that the top nodes identified by UniqueRank had positional attribute values similar to those selected by AttriRank; however, their force attributes tended to be significantly more extreme than those of the AttriRank-selected nodes.

Figures \ref{fig:toluene} and \ref{fig:naphthalene} illustrate these results for the top 100 and 500 nodes selected by both methods on toluene, aspirin, and naphthalene. Each row features a histogram for a single node attribute: purple indicates the distribution of attribute values over all nodes, blue indicates the distribution of attribute values over nodes selected by AttriRank, and red indicates the distribution of attribute values over nodes selected by UniqueRank.

Discussions with materials science experts indicate that the atoms with the most extreme force values, as selected by UniqueRank, tend to reside in unique chemical environments and exhibit distinct chemical properties. These results suggest that UniqueRank is effective at isolating atoms that may play especially interesting or critical roles in subsequent molecular dynamics simulations.
\begin{figure}[]
    \includegraphics[scale=.5]{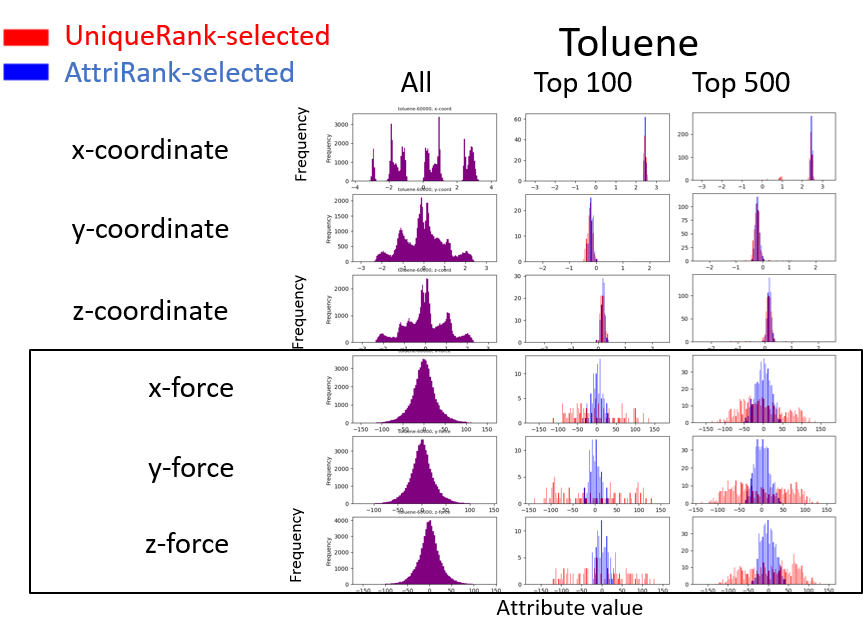}
    \caption{Distribution of attribute values on nodes that UniqueRank selects for toluene.}
    \label{fig:toluene}
\end{figure}

\begin{figure}[]
    \includegraphics[scale=.51]{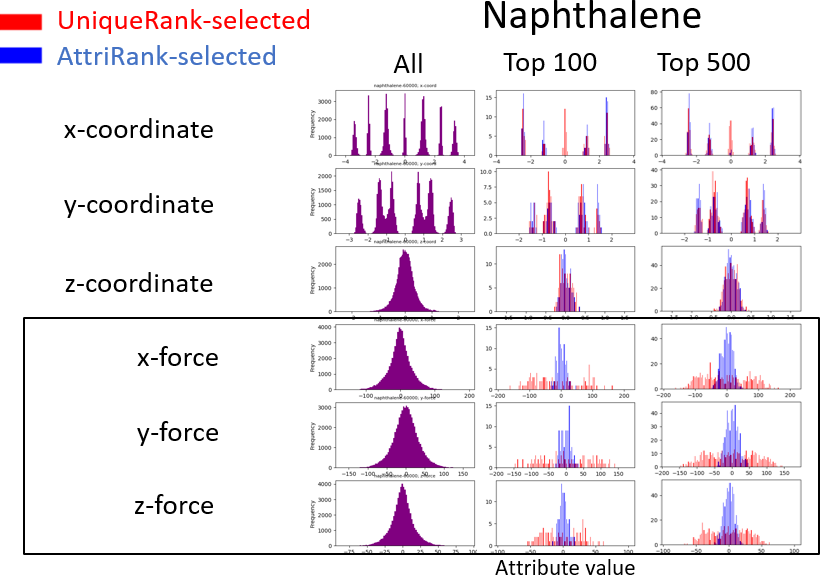}
    \caption{Distribution of attribute values on nodes that UniqueRank selects for naphthalene.}
    \label{fig:naphthalene}
\end{figure}

\section{Conclusion} \label{sec:conclusion} 

While identifying structurally important nodes whose removal can greatly disrupt network efficiency is a well-established problem, the related challenge of finding structurally important nodes with \textit{unique attributes}---whose removal and \textit{attempted replacement} can significantly affect efficiency---has received little attention. Nevertheless, attributed graphs are common in domains ranging from social networks and terrorist organizations to supply chains \cite{linvill2020troll,mcauley2012learning,williams2018real}, and real-world networks often aim to replace removed nodes with similar alternatives \cite{missaoui2013social}. Motivated by these practical considerations, we introduce \textbf{UniqueRank}, a Markov Chain-based method designed to identify nodes that both are structurally important and have unique attributes in their local neighborhood. We evaluate UniqueRank by measuring the reduction in network efficiency after removal and attempted replacement of its top-ranked nodes, and observe that it often corresponds to a greater efficiency loss compared to both advanced and traditional node importance ranking methods across a variety of real-world social, terrorist, and supply chain networks. Additionally, we apply UniqueRank to biomolecular structures, uncovering further insights into node uniqueness and its potential functional significance.

% or

\appendix
\section*{Subgraph Size and Efficiency: Pre- and Post-Modification Analysis}
\label{app:subgraphsizeplot}
Plots comparing the size of the 2-hop neighborhood of the removed node and network efficiency before and after node removal or modification, illustrated for the PIRA and Russian Trolls datasets. Results are shown for both UniqueRank and baseline methods at various thresholds.

\begin{figure}[]
    \includegraphics[scale=.7]{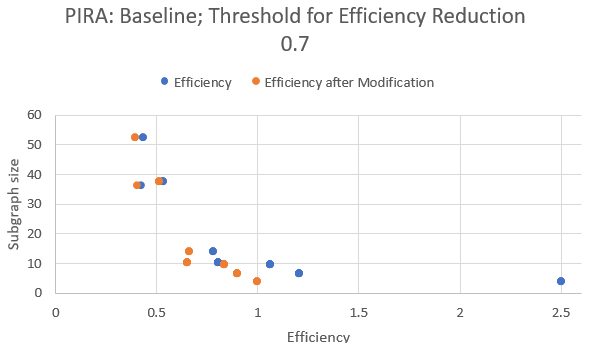}
    \caption{For PIRA, the plot of subgraph size and efficiency before and after node removal when the top 5 nodes are selected by the naive baseline (considering any baseline threshold in the range $\{0.7,...,1\}$), where the similarity threshold for replacement after removing a node is $0.7$. }
    \label{fig:pira-thresh-0.7}
\end{figure}

\begin{figure}[]
    \includegraphics[scale=.7]{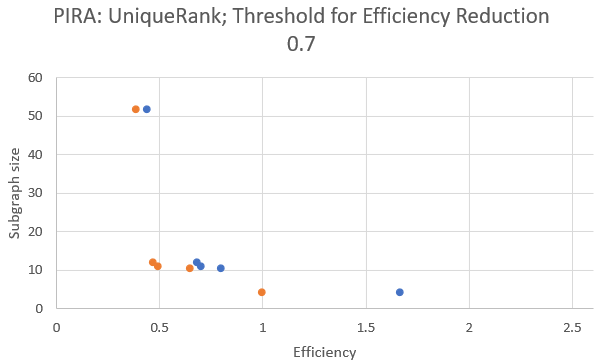}
    \caption{For PIRA, the plot of subgraph size and efficiency before and after node removal when the top 5 nodes are selected by UniqueRank, where the similarity threshold for replacement after removing a node is $0.7$. }
    \label{fig:pira-thresh-0.7-unique}
\end{figure}

\begin{figure}[]
    \includegraphics[scale=.7]{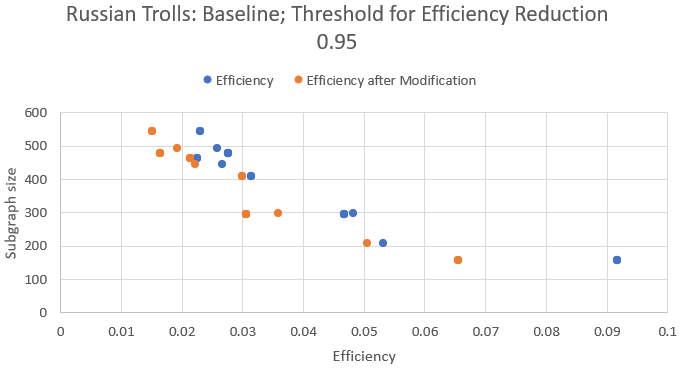}
    \caption{For Russian Troll Tweets, the plot of subgraph size and efficiency before and after node removal when the top 5 nodes are selected by the naive baseline (considering any baseline threshold in the range $\{0.7,...,1\}$), where the similarity threshold for replacement after removing a node is $0.95$. }
    \label{fig:rt-thresh-0.95}
\end{figure}

\begin{figure}[]
    \includegraphics[scale=.7]{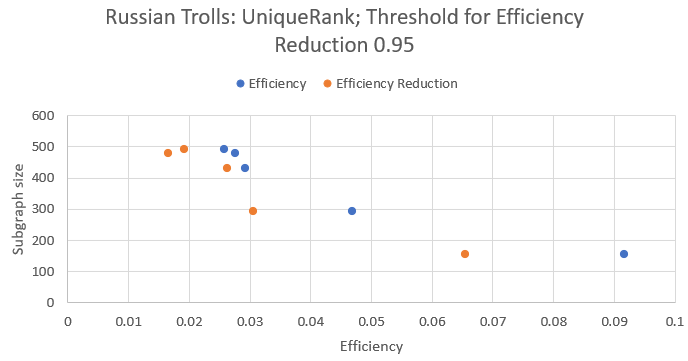}
    \caption{For Russian Troll Tweets, the plot of subgraph size and efficiency before and after node removal when the top 5 nodes are selected by UniqueRank, where the similarity threshold for replacement after removing a node is $0.95$. }
    \label{fig:rt-thresh-0.95-unique}
\end{figure}

\section*{Naive baseline: Average distances between a similar node and a removed node for more datasets}
\label{app:dist}

Additional results showing average distances between a similar node and a removed node for multiple baseline thresholds across more datasets in Tables~\ref{tab:distfb348}, \ref{tab:distsc37}, and \ref{tab:distfb3437}.

\begin{table}
\resizebox{\columnwidth}{!}{%
\begin{tabular}{cc|c|cccccc}
\toprule \multirow[t]{2}{*}{} & &  & \multicolumn{6}{|c}{Threshold for baseline} \\
\hline & & \makecell{Unique\\-Rank} & 0.95 & 0.9 & 0.85 & 0.8 & 0.75 & 0.7 \\
\midrule \multirow{6}{*}{\makecell{Threshold for \\efficiency \\reduction}}
 & 0.95 & 9.1 & 10 & 10 & 10 & 10 & 10 & 10 \\
 & 0.9  & 9.1 & 10 & 10 & 10 & 10 & 10 & 10 \\
 & 0.85 & 9.1 & 9.1 & 9.1 & 10 & 10 & 10 & 10 \\
 & 0.8  & 9.1 & 9.1 & 9.1 & 10 & 10 & 10 & 10 \\
 & 0.75 & 9.1 & 9.1 & 9.1 & 10 & 10 & 10 & 10 \\
 & 0.7  & 9.1 & 8.2 & 8.2 & 9.1 & 9.1 & 9.1 & 10 \\
\bottomrule
\end{tabular}}
\caption{For Facebook-348, the average distance between a removed node and a similar node passing a similarity threshold in the range $\{0.7,...,1\}$ over the top $10$ nodes selected by UniqueRank and the brute force baseline method, across thresholds in the range $\{0.7,...,1\}$.}
\label{tab:distfb348}
\end{table}

\begin{table}
\resizebox{\columnwidth}{!}{%
\begin{tabular}{cc|c|cccccc}
\toprule
\multirow[t]{2}{*}{} & & & \multicolumn{6}{|c}{Threshold for baseline} \\
\hline
& & \makecell{Unique\\-Rank} & 0.95 & 0.9 & 0.85 & 0.8 & 0.75 & 0.7 \\
\midrule
\multirow{6}{*}{\makecell{Threshold for \\ efficiency \\ reduction}}
 & 0.95 & 4.8 & 4.4 & 4.4 & 4.4 & 4.4 & 4.4 & 4.4 \\
 & 0.9  & 4.8 & 4.4 & 4.4 & 4.4 & 4.4 & 4.4 & 4.4 \\
 & 0.85 & 4.8 & 4.4 & 4.4 & 4.4 & 4.4 & 4.4 & 4.4 \\
 & 0.8  & 4.8 & 4.4 & 4.4 & 4.4 & 4.4 & 4.4 & 4.4 \\
 & 0.75 & 4.8 & 4.4 & 4.4 & 4.4 & 4.4 & 4.4 & 4.4 \\
 & 0.7  & 4.8 & 4.4 & 4.4 & 4.4 & 4.4 & 4.4 & 4.4 \\
\bottomrule
\end{tabular}}
\caption{For Supply-Chain/37, the average distance between a removed node and a similar node passing a similarity threshold in the range $\{0.7,\ldots,1\}$ over the top $5$ nodes selected by UniqueRank and the brute force baseline method, across thresholds in the range $\{0.7,\ldots,1\}$.}
\label{tab:distsc37}
\end{table}

\begin{table}
\resizebox{\columnwidth}{!}{%
\begin{tabular}{cc|c|cccccc}
\toprule \multirow[t]{2}{*}{} & &  & \multicolumn{6}{|c}{Threshold for baseline} \\
\hline
& & \makecell{Unique\\-Rank} & 0.95 & 0.9 & 0.85 & 0.8 & 0.75 & 0.7 \\
\midrule
\multirow{6}{*}{\makecell{Threshold for \\efficiency \\reduction}}
 & 0.95 & 10 & 9.3 & 9.3 & 10 & 10 & 10 & 10 \\
 & 0.9  & 10 & 9.3 & 9.3 & 10 & 10 & 10 & 10 \\
 & 0.85 & 10 & 7.6 & 7.6 & 9.3 & 9.3 & 10 & 10 \\
 & 0.8  & 10 & 7.6 & 7.6 & 9.3 & 9.3 & 10 & 10 \\
 & 0.75 & 10 & 7.4 & 7.4 & 9.2 & 9.2 & 10 & 10 \\
 & 0.7  & 10 & 7.4 & 7.4 & 9.2 & 9.2 & 10 & 10 \\
\bottomrule
\end{tabular}}
\caption{For Facebook-3437, the average distance between a removed node and a similar node passing a similarity threshold in the range $\{0.7, \ldots, 1\}$ over the top $10$ nodes selected by UniqueRank and the brute force baseline method, across thresholds in the range $\{0.7,\ldots,1\}$}
\label{tab:distfb3437}
\end{table}

\section*{Naive baseline: Efficiency reduction for more datasets}
\label{app:effreduc}

Additional results showing efficiency reduction for multiple baseline thresholds across more datasets in Tables~\ref{tab:efffb348}, \ref{tab:effsc37}, and \ref{tab:efffb3437}.

\begin{table*}
\resizebox{\textwidth}{!}{%
\begin{tabular}{cc|c|cccccc}
\toprule
\multirow[t]{2}{*}{} & & & \multicolumn{6}{|c}{Threshold for baseline} \\
\cline{3-9}
& & \makecell{UniqueRank \\ (Threshold- \\ independent)} & 0.95 & 0.9 & 0.85 & 0.8 & 0.75 & 0.7 \\
\midrule
\multirow{6}{*}{\makecell{Threshold for \\ efficiency \\ reduction}}
& 0.95 & $0.0979\pm 0.0449$ & $0.0334\pm 0.0157$ & $0.0334\pm 0.0157$ & $0.0331\pm 0.0158$ & $0.0331\pm 0.0158$ & $0.0331\pm 0.0158$ & $0.0319\pm 0.0160$ \\
& 0.9  & $0.0979\pm 0.0449$ & $0.0334\pm 0.0157$ & $0.0334\pm 0.0157$ & $0.0331\pm 0.0158$ & $0.0331\pm 0.0158$ & $0.0331\pm 0.0158$ & $0.0319\pm 0.0160$ \\
& 0.85 & $0.0979\pm 0.0449$ & $0.0333\pm 0.0158$ & $0.0333\pm 0.0158$ & $0.0333\pm 0.0158$ & $0.0333\pm 0.0158$ & $0.0333\pm 0.0158$ & $0.0320\pm 0.0160$ \\
& 0.8  & $0.0979\pm 0.0449$ & $0.0333\pm 0.0158$ & $0.0333\pm 0.0158$ & $0.0333\pm 0.0158$ & $0.0333\pm 0.0158$ & $0.0333\pm 0.0158$ & $0.0320\pm 0.0160$ \\
& 0.75 & $0.0979\pm 0.0449$ & $0.0333\pm 0.0158$ & $0.0333\pm 0.0158$ & $0.0333\pm 0.0158$ & $0.0333\pm 0.0158$ & $0.0333\pm 0.0158$ & $0.0320\pm 0.0160$ \\
& 0.7  & $0.0979\pm 0.0449$ & $0.0338\pm 0.0157$ & $0.0338\pm 0.0157$ & $0.0338\pm 0.0157$ & $0.0338\pm 0.0157$ & $0.0338\pm 0.0157$ & $0.0338\pm 0.0157$ \\
\bottomrule
\end{tabular}}
\caption{For Facebook-348, efficiency reduction over similarity threshold for replacement in the range $\{0.7,...,1\}$ and averaging over the top $10$ nodes selected by UniqueRank and the brute force baseline method, across thresholds in the range $\{0.7,...,1\}$.}
\label{tab:efffb348}
\end{table*}

\begin{table*}
\resizebox{\textwidth}{!}{%
\begin{tabular}{cc|c|cccccc}
\toprule
\multirow[t]{2}{*}{} & & & \multicolumn{6}{|c}{Threshold for baseline} \\
\cline{3-9}
& & \makecell{UniqueRank \\ (Threshold- \\ independent)} & 0.95 & 0.9 & 0.85 & 0.8 & 0.75 & 0.7 \\
\midrule
\multirow{6}{*}{\makecell{Threshold for \\ efficiency \\ reduction}}
& 0.95 & $0.5387\pm 0.0251$ & $0.5181\pm 0.0205$ & $0.5181\pm 0.0205$ & $0.5181\pm 0.0205$ & $0.5181\pm 0.0205$ & $0.5181\pm 0.0205$ & $0.5181\pm 0.0205$ \\
& 0.9  & $0.5387\pm 0.0251$ & $0.5181\pm 0.0205$ & $0.5181\pm 0.0205$ & $0.5181\pm 0.0205$ & $0.5181\pm 0.0205$ & $0.5181\pm 0.0205$ & $0.5181\pm 0.0205$ \\
& 0.85 & $0.5387\pm 0.0251$ & $0.5181\pm 0.0205$ & $0.5181\pm 0.0205$ & $0.5181\pm 0.0205$ & $0.5181\pm 0.0205$ & $0.5181\pm 0.0205$ & $0.5181\pm 0.0205$ \\
& 0.8  & $0.5387\pm 0.0251$ & $0.5181\pm 0.0205$ & $0.5181\pm 0.0205$ & $0.5181\pm 0.0205$ & $0.5181\pm 0.0205$ & $0.5181\pm 0.0205$ & $0.5181\pm 0.0205$ \\
& 0.75 & $0.5387\pm 0.0251$ & $0.5181\pm 0.0205$ & $0.5181\pm 0.0205$ & $0.5181\pm 0.0205$ & $0.5181\pm 0.0205$ & $0.5181\pm 0.0205$ & $0.5181\pm 0.0205$ \\
& 0.7  & $0.5387\pm 0.0251$ & $0.5181\pm 0.0205$ & $0.5181\pm 0.0205$ & $0.5181\pm 0.0205$ & $0.5181\pm 0.0205$ & $0.5181\pm 0.0205$ & $0.5181\pm 0.0205$ \\
\bottomrule
\end{tabular}}
\caption{For Supply-Chain/37, efficiency reduction over similarity threshold for replacement in the range $\{0.7,...,1\}$ and averaging over the top $5$ nodes selected by UniqueRank and the brute force baseline method, across thresholds in the range $\{0.7,...,1\}$.}
\label{tab:effsc37}
\end{table*}

\begin{table*}
\resizebox{\textwidth}{!}{%
\begin{tabular}{cc|c|cccccc}
\toprule
\multirow[t]{2}{*}{} & & & \multicolumn{6}{|c}{Threshold for baseline} \\
\cline{3-9}
& & \makecell{UniqueRank \\ (Threshold- \\ independent)} & 0.95 & 0.9 & 0.85 & 0.8 & 0.75 & 0.7 \\
\midrule
\multirow{6}{*}{\makecell{Threshold for \\ efficiency \\ reduction}}
& 0.95 & $0.0177\pm 0.0017$ & $0.0217\pm 0.0031$ & $0.0217\pm 0.0031$ & $0.0198\pm 0.0035$ & $0.0198\pm 0.0035$ & $0.0189\pm 0.0037$ & $0.0189\pm 0.0037$ \\
& 0.9  & $0.0177\pm 0.0017$ & $0.0217\pm 0.0031$ & $0.0217\pm 0.0031$ & $0.0198\pm 0.0035$ & $0.0198\pm 0.0035$ & $0.0189\pm 0.0037$ & $0.0189\pm 0.0037$ \\
& 0.85 & $0.0177\pm 0.0017$ & $0.0218\pm 0.0033$ & $0.0218\pm 0.0033$ & $0.0218\pm 0.0033$ & $0.0218\pm 0.0033$ & $0.0209\pm 0.0036$ & $0.0209\pm 0.0036$ \\
& 0.8  & $0.0177\pm 0.0017$ & $0.0218\pm 0.0033$ & $0.0218\pm 0.0033$ & $0.0218\pm 0.0033$ & $0.0218\pm 0.0033$ & $0.0209\pm 0.0036$ & $0.0209\pm 0.0036$ \\
& 0.75 & $0.0177\pm 0.0017$ & $0.0214\pm 0.0034$ & $0.0214\pm 0.0034$ & $0.0214\pm 0.0034$ & $0.0214\pm 0.0034$ & $0.0214\pm 0.0034$ & $0.0214\pm 0.0034$ \\
& 0.7  & $0.0177\pm 0.0017$ & $0.0214\pm 0.0034$ & $0.0214\pm 0.0034$ & $0.0214\pm 0.0034$ & $0.0214\pm 0.0034$ & $0.0214\pm 0.0034$ & $0.0214\pm 0.0034$ \\
\bottomrule
\end{tabular}}
\caption{For Facebook-3437, efficiency reduction over similarity threshold for replacement in the range $\{0.7,...,1\}$ and averaging over the top $5$ nodes selected by UniqueRank and the brute force baseline method, across thresholds in the range $\{0.7,...,1\}$.}
\label{tab:efffb3437}
\end{table*}

%\appendices

\ifCLASSOPTIONcompsoc
  \section*{Acknowledgments}
\else
 
  \section*{Acknowledgment}
 \fi

The authors would like to thank Danielle Sullivan, Sheila Alemany Blanco, Giselle Zeno, Michael Coury, and Evan Young from MIT Lincoln Laboratory for their feedback.

\ifCLASSOPTIONcaptionsoff
  \newpage
\fi



\bibliographystyle{IEEEtran}

%\bibliography{IEEEabrv,../bib/paper,references}
\bibliography{references}

\begin{thebibliography}{10}
\providecommand{\url}[1]{#1}
\csname url@samestyle\endcsname
\providecommand{\newblock}{\relax}
\providecommand{\bibinfo}[2]{#2}
\providecommand{\BIBentrySTDinterwordspacing}{\spaceskip=0pt\relax}
\providecommand{\BIBentryALTinterwordstretchfactor}{4}
\providecommand{\BIBentryALTinterwordspacing}{\spaceskip=\fontdimen2\font plus
\BIBentryALTinterwordstretchfactor\fontdimen3\font minus
  \fontdimen4\font\relax}
\providecommand{\BIBforeignlanguage}[2]{{%
\expandafter\ifx\csname l@#1\endcsname\relax
\typeout{** WARNING: IEEEtran.bst: No hyphenation pattern has been}%
\typeout{** loaded for the language `#1'. Using the pattern for}%
\typeout{** the default language instead.}%
\else
\language=\csname l@#1\endcsname
\fi
#2}}
\providecommand{\BIBdecl}{\relax}
\BIBdecl

\bibitem{ZAREIE2018hierarchical}
A.~Zareie and A.~Sheikhahmadi, ``A hierarchical approach for influential node
  ranking in complex social networks,'' \emph{Expert Systems with
  Applications}, vol.~93, pp. 200--211, 2018.

\bibitem{WANG2017ranking}
Z.~Wang, C.~Du, J.~Fan, and Y.~Xing, ``Ranking influential nodes in social
  networks based on node position and neighborhood,'' \emph{Neurocomputing},
  vol. 260, pp. 466--477, 2017.

\bibitem{cheng2011virtual}
X.~Cheng, S.~Su, Z.~Zhang, H.~Wang, F.~Yang, Y.~Luo, and J.~Wang, ``Virtual
  network embedding through topology-aware node ranking,'' \emph{SIGCOMM
  Comput. Commun. Rev.}, vol.~41, no.~2, p. 38–47, Apr. 2011.

\bibitem{SALAVATI2019ranking}
C.~Salavati, A.~Abdollahpouri, and Z.~Manbari, ``Ranking nodes in complex
  networks based on local structure and improving closeness centrality,''
  \emph{Neurocomputing}, vol. 336, pp. 36--45, 2019, advances in Graph
  Algorithm and Applications.

\bibitem{Page1998PageRank}
\BIBentryALTinterwordspacing
L.~Page, S.~Brin, R.~Motwani, and T.~Winograd, ``The pagerank citation ranking:
  Bringing order to the web.'' Stanford InfoLab, Technical Report 1999-66,
  November 1999, previous number = SIDL-WP-1999-0120. [Online]. Available:
  \url{http://ilpubs.stanford.edu:8090/422/}
\BIBentrySTDinterwordspacing

\bibitem{brin1998anatomy}
S.~Brin and L.~Page, ``The anatomy of a large-scale hypertextual web search
  engine,'' in \emph{Proceedings of the Seventh International Conference on
  World Wide Web 7}, ser. WWW7.\hskip 1em plus 0.5em minus 0.4em\relax NLD:
  Elsevier Science Publishers B. V., 1998, p. 107–117.

\bibitem{hsu2017unsupervised}
C.-C. Hsu, Y.-A. Lai, W.-H. Chen, M.-H. Feng, and S.-D. Lin, ``Unsupervised
  ranking using graph structures and node attributes,'' in \emph{Proceedings of
  the Tenth ACM International Conference on Web Search and Data Mining}, ser.
  WSDM '17.\hskip 1em plus 0.5em minus 0.4em\relax New York, NY, USA:
  Association for Computing Machinery, 2017, p. 771–779.

\bibitem{benyahia2015centrality}
O.~Benyahia and C.~Largeron, ``Centrality for graphs with numerical
  attributes,'' in \emph{2015 IEEE/ACM International Conference on Advances in
  Social Networks Analysis and Mining (ASONAM)}, 2015, pp. 1348--1353.

\bibitem{missaoui2013social}
R.~Missaoui, E.~Negre, D.~Anggraini, and J.~Vaillancourt, ``Social network
  restructuring after a node removal,'' \emph{Int. J. Web Eng. Technol.},
  vol.~8, no.~1, p. 4–26, Mar. 2013.

\bibitem{latora2001efficient}
V.~Latora and M.~Marchiori, ``Efficient behavior of small-world networks,''
  \emph{Phys. Rev. Lett.}, vol.~87, p. 198701, Oct 2001.

\bibitem{Rodrigues2019network}
F.~A. Rodrigues, \emph{Network Centrality: An Introduction}.\hskip 1em plus
  0.5em minus 0.4em\relax Cham: Springer International Publishing, 2019, pp.
  177--196.

\bibitem{Das2018-aj}
K.~Das, S.~Samanta, and M.~Pal, ``Study on centrality measures in social
  networks: a survey,'' \emph{Social Network Analysis and Mining}, vol.~8,
  no.~1, p.~13, Feb. 2018.

\bibitem{QU2023gnr}
H.~Qu, Y.-R. Song, R.~Li, and M.~Li, ``{GNR}: A universal and efficient node
  ranking model for various tasks based on graph neural networks,''
  \emph{Physica A: Statistical Mechanics and its Applications}, vol. 632, p.
  129339, 2023.

\bibitem{maurya2021graph}
S.~K. Maurya, X.~Liu, and T.~Murata, ``Graph neural networks for fast node
  ranking approximation,'' \emph{ACM Trans. Knowl. Discov. Data}, vol.~15,
  no.~5, May 2021.

\bibitem{li2024novel}
S.~Li, D.~Ji, J.~Wang, and W.~Xia, ``A novel gnn-based node importance ranking
  method in a heterogeneous network,'' in \emph{Proceedings of 2023 7th Chinese
  Conference on Swarm Intelligence and Cooperative Control}, Y.~Hua, Y.~Liu,
  and L.~Han, Eds.\hskip 1em plus 0.5em minus 0.4em\relax Singapore: Springer
  Nature Singapore, 2024, pp. 246--257.

\bibitem{park2019estimating}
N.~Park, A.~Kan, X.~L. Dong, T.~Zhao, and C.~Faloutsos, ``Estimating node
  importance in knowledge graphs using graph neural networks,'' in
  \emph{Proceedings of the 25th ACM SIGKDD International Conference on
  Knowledge Discovery \& Data Mining}, ser. KDD '19.\hskip 1em plus 0.5em minus
  0.4em\relax New York, NY, USA: Association for Computing Machinery, 2019, p.
  596–606.

\bibitem{ergashev2023learning}
U.~Ergashev, E.~Dragut, and W.~Meng, ``Learning to rank resources with {GNN},''
  in \emph{Proceedings of the ACM Web Conference 2023}, ser. WWW '23.\hskip 1em
  plus 0.5em minus 0.4em\relax New York, NY, USA: Association for Computing
  Machinery, 2023, p. 3247–3256.

\bibitem{he2022gnnrank}
Y.~He, Q.~Gan, D.~Wipf, G.~D. Reinert, J.~Yan, and M.~Cucuringu, ``{GNNR}ank:
  Learning global rankings from pairwise comparisons via directed graph neural
  networks,'' in \emph{Proceedings of the 39th International Conference on
  Machine Learning}, ser. Proceedings of Machine Learning Research,
  K.~Chaudhuri, S.~Jegelka, L.~Song, C.~Szepesvari, G.~Niu, and S.~Sabato,
  Eds., vol. 162.\hskip 1em plus 0.5em minus 0.4em\relax PMLR, 17--23 Jul 2022,
  pp. 8581--8612.

\bibitem{bellingeri2020link}
M.~Bellingeri, D.~Bevacqua, F.~Scotognella, R.~Alfieri, Q.~Nguyen,
  D.~Montepietra, and D.~Cassi, ``Link and node removal in real social
  networks: A review,'' \emph{Frontiers in Physics}, vol.~8, 2020.

\bibitem{Amancio_2015}
D.~R. Amancio, O.~N. Oliveira, and L.~da~F~Costa, ``Robustness of community
  structure to node removal,'' \emph{Journal of Statistical Mechanics: Theory
  and Experiment}, vol. 2015, no.~3, p. P03003, mar 2015.

\bibitem{Albert2000-cm}
R.~Albert, H.~Jeong, and A.-L. Barab{\'a}si, ``Error and attack tolerance of
  complex networks,'' \emph{Nature}, vol. 406, no. 6794, pp. 378--382, Jul.
  2000.

\bibitem{CallawayNewmanStrogatzWatts+2006+510+513}
D.~S. Callaway, M.~Newman, S.~H. Strogatz, and D.~J. Watts, \emph{Network
  Robustness and Fragility: Percolation on Random Graphs}.\hskip 1em plus 0.5em
  minus 0.4em\relax Princeton: Princeton University Press, 2006, pp. 510--513.

\bibitem{Albert_2002}
R.~Albert and A.-L. Barabási, ``Statistical mechanics of complex networks,''
  \emph{Reviews of Modern Physics}, vol.~74, no.~1, p. 47–97, Jan. 2002.

\bibitem{holme2002attack}
P.~Holme, B.~J. Kim, C.~N. Yoon, and S.~K. Han, ``Attack vulnerability of
  complex networks,'' \emph{Phys. Rev. E}, vol.~65, p. 056109, May 2002.

\bibitem{chen2013node}
P.-Y. Chen and A.~O. Hero, ``Node removal vulnerability of the largest
  component of a network,'' in \emph{2013 IEEE Global Conference on Signal and
  Information Processing}, 2013, pp. 587--590.

\bibitem{JAHANPOUR20133458}
E.~Jahanpour and X.~Chen, ``Analysis of complex network performance and
  heuristic node removal strategies,'' \emph{Communications in Nonlinear
  Science and Numerical Simulation}, vol.~18, no.~12, pp. 3458--3468, 2013.

\bibitem{Boldi2013-wp}
P.~Boldi, M.~Rosa, and S.~Vigna, ``Robustness of social and web graphs to node
  removal,'' \emph{Social Network Analysis and Mining}, vol.~3, no.~4, pp.
  829--842, Dec. 2013.

\bibitem{Iyer2013-fm}
S.~Iyer, T.~Killingback, B.~Sundaram, and Z.~Wang,
  ``\BIBforeignlanguage{en}{Attack robustness and centrality of complex
  networks},'' \emph{\BIBforeignlanguage{en}{PLoS One}}, vol.~8, no.~4, p.
  e59613, Apr. 2013.

\bibitem{BELLINGERI2014174}
M.~Bellingeri, D.~Cassi, and S.~Vincenzi, ``Efficiency of attack strategies on
  complex model and real-world networks,'' \emph{Physica A: Statistical
  Mechanics and its Applications}, vol. 414, pp. 174--180, 2014.

\bibitem{gallos2005stability}
L.~K. Gallos, R.~Cohen, P.~Argyrakis, A.~Bunde, and S.~Havlin, ``Stability and
  topology of scale-free networks under attack and defense strategies,''
  \emph{Phys. Rev. Lett.}, vol.~94, p. 188701, May 2005.

\bibitem{NIE2015248}
T.~Nie, Z.~Guo, K.~Zhao, and Z.-M. Lu, ``New attack strategies for complex
  networks,'' \emph{Physica A: Statistical Mechanics and its Applications},
  vol. 424, pp. 248--253, 2015.

\bibitem{Tian2017-en}
L.~Tian, A.~Bashan, D.-N. Shi, and Y.-Y. Liu, ``Articulation points in complex
  networks,'' \emph{Nature Communications}, vol.~8, no.~1, p. 14223, Jan. 2017.

\bibitem{NGUYEN2019121561}
Q.~Nguyen, H.~Pham, D.~Cassi, and M.~Bellingeri, ``Conditional attack strategy
  for real-world complex networks,'' \emph{Physica A: Statistical Mechanics and
  its Applications}, vol. 530, p. 121561, 2019.

\bibitem{Morone2015-cg}
F.~Morone and H.~A. Makse, ``Influence maximization in complex networks through
  optimal percolation,'' \emph{Nature}, vol. 524, no. 7563, pp. 65--68, Aug.
  2015.

\bibitem{Schneider_2012inverse}
C.~M. Schneider, T.~Mihaljev, and H.~J. Herrmann, ``Inverse targeting —an
  effective immunization strategy,'' \emph{EPL (Europhysics Letters)}, vol.~98,
  no.~4, p. 46002, May 2012.

\bibitem{P.Holme_2004}
P.~Holme, ``Efficient local strategies for vaccination and network attack,''
  \emph{Europhysics Letters}, vol.~68, no.~6, p. 908, nov 2004.

\bibitem{Wang_2015immunity}
Z.~Wang, D.-W. Zhao, L.~Wang, G.-Q. Sun, and Z.~Jin, ``Immunity of multiplex
  networks via acquaintance vaccination,'' \emph{Europhysics Letters}, vol.
  112, no.~4, p. 48002, dec 2015.

\bibitem{schneider2011suppressing}
C.~M. Schneider, T.~Mihaljev, S.~Havlin, and H.~J. Herrmann, ``Suppressing
  epidemics with a limited amount of immunization units,'' \emph{Phys. Rev. E},
  vol.~84, p. 061911, Dec 2011.

\bibitem{bellingeri2015optimization}
M.~Bellingeri, E.~Agliari, and D.~Cassi, ``Optimization strategies with
  resource scarcity. from immunization of networks to the traveling salesman
  problem,'' \emph{Modern Physics Letters B}, 09 2015.

\bibitem{chen2008finding}
Y.~Chen, G.~Paul, S.~Havlin, F.~Liljeros, and H.~E. Stanley, ``Finding a better
  immunization strategy,'' \emph{Phys. Rev. Lett.}, vol. 101, p. 058701, Jul
  2008.

\bibitem{Hadidjojo2011-equal}
J.~Hadidjojo and S.~A. Cheong, ``\BIBforeignlanguage{en}{Equal graph
  partitioning on estimated infection network as an effective epidemic
  mitigation measure},'' \emph{\BIBforeignlanguage{en}{PLoS One}}, vol.~6,
  no.~7, p. e22124, Jul. 2011.

\bibitem{Pan_2012strength}
R.~K. Pan and J.~Saramäki, ``The strength of strong ties in scientific
  collaboration networks,'' \emph{Europhysics Letters}, vol.~97, no.~1, p.
  18007, jan 2012.

\bibitem{AGRESTE201630}
S.~Agreste, S.~Catanese, P.~{De Meo}, E.~Ferrara, and G.~Fiumara, ``Network
  structure and resilience of mafia syndicates,'' \emph{Information Sciences},
  vol. 351, pp. 30--47, 2016.

\bibitem{requiao2018topology}
B.~Requiao Da~Cunha and S.~Gonçalves, ``Topology, robustness, and structural
  controllability of the brazilian federal police criminal intelligence
  network,'' \emph{Applied Network Science}, vol.~3, 08 2018.

\bibitem{petersen2011node}
R.~R. Petersen, C.~J. Rhodes, and U.~K. Wiil, ``Node removal in criminal
  networks,'' in \emph{2011 European Intelligence and Security Informatics
  Conference}, 2011, pp. 360--365.

\bibitem{qi2013terrorist}
X.~Qi, R.~Duval, K.~Christensen, E.~Fuller, A.~Spahiu, Q.~Wu, Y.~Wu, W.~Tang,
  and C.-q. Zhang, ``Terrorist networks, network energy and node removal: A new
  measure of centrality based on laplacian energy,'' \emph{Social Networking},
  vol.~02, pp. 19--31, 01 2013.

\bibitem{zhao2019supply}
K.~Zhao, K.~Scheibe, J.~Blackhurst, and A.~Kumar, ``Supply chain network
  robustness against disruptions: Topological analysis, measurement, and
  optimization,'' \emph{IEEE Transactions on Engineering Management}, vol.~66,
  no.~1, pp. 127--139, 2019.

\bibitem{Thadakamaila2004survivability}
T.~HP, U.~Raghavan, S.~Kumara, and R.~Albert, ``Survivability of
  multiagent-based supply networks: a topological perspect,'' \emph{IEEE
  Intelligent Systems}, vol.~19, no.~5, pp. 24--31, 2004.

\bibitem{zhao2011analyzing}
K.~Zhao, A.~Kumar, T.~P. Harrison, and J.~Yen, ``Analyzing the resilience of
  complex supply network topologies against random and targeted disruptions,''
  \emph{IEEE Systems Journal}, vol.~5, no.~1, pp. 28--39, 2011.

\bibitem{tian2021research}
Y.~Tian, Y.~Shi, X.~Shi, M.~Li, and M.~Zhang, ``Research on supply chain
  network resilience considering the exit and reselection of enterprises,''
  \emph{IEEE Access}, vol.~9, pp. 91\,265--91\,281, 2021.

\bibitem{Shi2012ADM}
B.~Shi and J.~Liu, ``A decentralized mechanism for improving the functional
  robustness of distribution networks,'' \emph{IEEE Transactions on Systems,
  Man, and Cybernetics, Part B (Cybernetics)}, vol.~42, pp. 1369--1382, 2012.

\bibitem{zhao2011achieving}
K.~Zhao, A.~Kumar, and J.~Yen, ``Achieving high robustness in supply
  distribution networks by rewiring,'' \emph{IEEE Transactions on Engineering
  Management}, vol.~58, no.~2, pp. 347--362, 2011.

\bibitem{bimpikis2019supply}
K.~Bimpikis, O.~Candogan, and S.~Ehsani, ``{Supply Disruptions and Optimal
  Network Structures},'' \emph{Management Science}, vol.~65, no.~12, pp.
  5504--5517, December 2019.

\bibitem{Chen2023-ue}
H.~Chen, G.~Chen, Q.~Zhang, and X.~Zhang, ``\BIBforeignlanguage{en}{Analysis of
  network disruption evolution of chinese fresh cold chain under {COVID-19}},''
  \emph{\BIBforeignlanguage{en}{PLoS One}}, vol.~18, no.~1, p. e0278697, Jan.
  2023.

\bibitem{zhu2012disruptions}
S.~Zhu and D.~M. Levinson, ``Disruptions to transportation networks: A
  review,'' in \emph{Network Reliability in Practice}, D.~M. Levinson, H.~X.
  Liu, and M.~Bell, Eds.\hskip 1em plus 0.5em minus 0.4em\relax New York, NY:
  Springer New York, 2012, pp. 5--20.

\bibitem{HE2012modeling}
X.~He and H.~X. Liu, ``Modeling the day-to-day traffic evolution process after
  an unexpected network disruption,'' \emph{Transportation Research Part B:
  Methodological}, vol.~46, no.~1, pp. 50--71, 2012.

\bibitem{Landherr2010-fg}
A.~Landherr, B.~Friedl, and J.~Heidemann, ``A critical review of centrality
  measures in social networks,'' \emph{Business \& Information Systems
  Engineering}, vol.~2, no.~6, pp. 371--385, Dec. 2010.

\bibitem{BONACICH2007555}
P.~Bonacich, ``Some unique properties of eigenvector centrality,'' \emph{Social
  Networks}, vol.~29, no.~4, pp. 555--564, 2007.

\bibitem{grassi2007some}
R.~Grassi, S.~Stefani, and A.~Torriero, ``Some new results on the eigenvector
  centrality,'' \emph{The Journal of Mathematical Sociology}, vol.~31, no.~3,
  pp. 237--248, 2007.

\bibitem{Kleinberg99hits}
J.~M. Kleinberg, ``Authoritative sources in a hyperlinked environment,''
  \emph{J. ACM}, vol.~46, no.~5, pp. 604--632, 1999.

\bibitem{yang2019node}
Y.~Yang, L.~Yu, Z.~Zhou, Y.~Chen, and T.~Kou, ``Node importance ranking in
  complex networks based on multicriteria decision making,'' \emph{Mathematical
  Problems in Engineering}, vol. 2019, pp. 1--12, 01 2019.

\bibitem{NEWMAN200539}
M.~J. Newman, ``A measure of betweenness centrality based on random walks,''
  \emph{Social Networks}, vol.~27, no.~1, pp. 39--54, 2005.

\bibitem{bandes2001faster}
U.~Brandes, ``A faster algorithm for betweenness centrality*,'' \emph{The
  Journal of Mathematical Sociology}, vol.~25, no.~2, pp. 163--177, 2001.

\bibitem{chen2021node}
Q.~Chen, Q.~Chen, X.~Zhao, W.~Sun, and C.~Wang, ``A node importance ranking
  method based on the rate of network entropy changes,'' in \emph{2021 17th
  International Conference on Mobility, Sensing and Networking (MSN)}, 2021,
  pp. 32--39.

\bibitem{Liu2023structure}
S.~Liu and H.~Gao, ``\BIBforeignlanguage{en}{The structure {Entropy-Based} node
  importance ranking method for graph data},''
  \emph{\BIBforeignlanguage{en}{Entropy (Basel)}}, vol.~25, no.~6, Jun. 2023.

\bibitem{Yu2022-identifying}
Y.~Yu, B.~Zhou, L.~Chen, T.~Gao, and J.~Liu,
  ``\BIBforeignlanguage{en}{Identifying important nodes in complex networks
  based on node propagation entropy},'' \emph{\BIBforeignlanguage{en}{Entropy
  (Basel)}}, vol.~24, no.~2, Feb. 2022.

\bibitem{Lu2016-ma}
L.~L{\"u}, T.~Zhou, Q.-M. Zhang, and H.~E. Stanley, ``The h-index of a network
  node and its relation to degree and coreness,'' \emph{Nature Communications},
  vol.~7, no.~1, p. 10168, Jan. 2016.

\bibitem{Sheikhahmadi2022multi}
A.~Sheikhahmadi, F.~Veisi, A.~Sheikhahmadi, and S.~Mohammadimajd,
  ``\BIBforeignlanguage{en}{A multi-attribute method for ranking influential
  nodes in complex networks},'' \emph{\BIBforeignlanguage{en}{PLoS One}},
  vol.~17, no.~11, p. e0278129, Nov. 2022.

\bibitem{LIU2015importance}
Z.~Liu, C.~Jiang, J.~Wang, and H.~Yu, ``The node importance in actual complex
  networks based on a multi-attribute ranking method,'' \emph{Knowledge-Based
  Systems}, vol.~84, pp. 56--66, 2015.

\bibitem{yang2018multi}
P.~Yang, G.~Xu, and H.~Chen, ``Multi-attribute ranking method for identifying
  key nodes in complex networks based on gra,'' \emph{International Journal of
  Modern Physics B}, vol.~32, no.~32, p. 1850363, 2018.

\bibitem{gong2016virtual}
S.~Gong, J.~Chen, S.~Zhao, and Q.~Zhu, ``Virtual network embedding with
  multi-attribute node ranking based on topsis,'' \emph{KSII Transactions on
  Internet and Information Systems (TIIS)}, vol.~10, pp. 522--541, 02 2016.

\bibitem{cao2018novel}
H.~Cao, L.~Yang, and H.~Zhu, ``Novel node-ranking approach and multiple
  topology attributes-based embedding algorithm for single-domain virtual
  network embedding,'' \emph{IEEE Internet of Things Journal}, vol.~5, no.~1,
  pp. 108--120, 2018.

\bibitem{sanchez2014local}
P.~I. S\'{a}nchez, E.~M\"{u}ller, O.~Irmler, and K.~B\"{o}hm, ``Local context
  selection for outlier ranking in graphs with multiple numeric node
  attributes,'' in \emph{Proceedings of the 26th International Conference on
  Scientific and Statistical Database Management}, ser. SSDBM '14.\hskip 1em
  plus 0.5em minus 0.4em\relax New York, NY, USA: Association for Computing
  Machinery, 2014.

\bibitem{lopes2017use}
R.~Lopes~de Andrade and L.~Rego, ``The use of nodes attributes in social
  network analysis with an application to an international trade network,''
  \emph{Physica A: Statistical Mechanics and its Applications}, vol. 491, 09
  2017.

\bibitem{dearruda2024assigningentitiesteamshypergraph}
\BIBentryALTinterwordspacing
G.~F. de~Arruda, W.~He, N.~Heydaribeni, T.~Javidi, Y.~Moreno, and
  T.~Eliassi-Rad, ``Assigning entities to teams as a hypergraph discovery
  problem,'' 2024. [Online]. Available: \url{https://arxiv.org/abs/2403.04063}
\BIBentrySTDinterwordspacing

\bibitem{juarez2021comprehensive}
J.~Ju\'{a}rez, C.~P. Santos, and C.~A. Brizuela, ``A comprehensive review and a
  taxonomy proposal of team formation problems,'' \emph{ACM Comput. Surv.},
  vol.~54, no.~7, Jul. 2021.

\bibitem{Romanini2021-privacy}
D.~Romanini, S.~Lehmann, and M.~Kivel{\"a}, ``Privacy and uniqueness of
  neighborhoods in social networks,'' \emph{Scientific Reports}, vol.~11,
  no.~1, p. 20104, Oct. 2021.

\bibitem{Häggström_2002irreducible}
O.~Häggström, \emph{Irreducible and aperiodic Markov chains}, ser. London
  Mathematical Society Student Texts.\hskip 1em plus 0.5em minus 0.4em\relax
  Cambridge University Press, 2002, p. 23–27.

\bibitem{linvill2020troll}
D.~L. Linvill and P.~L. Warren, ``Troll factories: Manufacturing specialized
  disinformation on twitter,'' \emph{Political Communication}, vol.~37, no.~4,
  pp. 447--467, 2020.

\bibitem{mcauley2012learning}
J.~McAuley and J.~Leskovec, ``Learning to discover social circles in ego
  networks,'' in \emph{Proceedings of the 25th International Conference on
  Neural Information Processing Systems - Volume 1}, ser. NIPS'12.\hskip 1em
  plus 0.5em minus 0.4em\relax Curran Associates Inc., 2012, p. 539–547.

\bibitem{yang2020scaling}
R.~Yang, J.~Shi, X.~Xiao, Y.~Yang, J.~Liu, and S.~S. Bhowmick, ``Scaling
  attributed network embedding to massive graphs,'' \emph{Proceedings of the
  VLDB Endowment}, vol.~14, no.~1, pp. 37--49, 2021.

\bibitem{williams2018real}
S.~P. Willems, ``Data set---real-world multiechelon supply chains used for
  inventory optimization,'' \emph{Manufacturing \& Service Operations
  Management}, vol.~10, no.~1, p. 19–23, Jan. 2008.

\bibitem{Borgwardt2005protein}
K.~M. Borgwardt, C.~S. Ong, S.~Schönauer, S.~V.~N. Vishwanathan, A.~J. Smola,
  and H.-P. Kriegel, ``Protein function prediction via graph kernels,''
  \emph{Bioinformatics}, vol.~21, no. suppl\_1, pp. i47--i56, 06 2005.

\bibitem{Schomburg2004-vm}
I.~Schomburg, A.~Chang, C.~Ebeling, M.~Gremse, C.~Heldt, G.~Huhn, and
  D.~Schomburg, ``\BIBforeignlanguage{en}{{BRENDA}, the enzyme database:
  updates and major new developments},'' \emph{\BIBforeignlanguage{en}{Nucleic
  Acids Res}}, vol.~32, no. Database issue, pp. D431--3, Jan. 2004.

\bibitem{Dobson2003-zz}
P.~D. Dobson and A.~J. Doig, ``\BIBforeignlanguage{en}{Distinguishing enzyme
  structures from non-enzymes without alignments},''
  \emph{\BIBforeignlanguage{en}{J Mol Biol}}, vol. 330, no.~4, pp. 771--783,
  Jul. 2003.

\bibitem{chmiela2017machine}
S.~Chmiela, A.~Tkatchenko, H.~E. Sauceda, I.~Poltavsky, K.~T. Schütt, and
  K.-R. Müller, ``Machine learning of accurate energy-conserving molecular
  force fields,'' \emph{Science Advances}, vol.~3, no.~5, p. e1603015, 2017.

\bibitem{Neumann2016-jb}
M.~Neumann, R.~Garnett, C.~Bauckhage, and K.~Kersting, ``Propagation kernels:
  efficient graph kernels from propagated information,'' \emph{Machine
  Learning}, vol. 102, no.~2, pp. 209--245, Feb. 2016.

\bibitem{riesen2008iam}
K.~Riesen and H.~Bunke, ``{IAM} graph database repository for graph based
  pattern recognition and machine learning,'' in \emph{Structural, Syntactic,
  and Statistical Pattern Recognition}, N.~da~Vitoria~Lobo, T.~Kasparis,
  F.~Roli, J.~T. Kwok, M.~Georgiopoulos, G.~C. Anagnostopoulos, and M.~Loog,
  Eds.\hskip 1em plus 0.5em minus 0.4em\relax Berlin, Heidelberg: Springer
  Berlin Heidelberg, 2008, pp. 287--297.

\bibitem{meng2019coembedding}
Z.~Meng, S.~Liang, H.~Bao, and X.~Zhang, ``Co-embedding attributed networks,''
  in \emph{Proceedings of the Twelfth ACM International Conference on Web
  Search and Data Mining}, ser. WSDM '19.\hskip 1em plus 0.5em minus
  0.4em\relax New York, NY, USA: Association for Computing Machinery, 2019, p.
  393–401.

\bibitem{Yang2013community}
J.~Yang, J.~McAuley, and J.~Leskovec, ``Community detection in networks with
  node attributes,'' in \emph{2013 IEEE 13th International Conference on Data
  Mining}.\hskip 1em plus 0.5em minus 0.4em\relax IEEE, Dec. 2013.

\bibitem{zhang2018anrl}
Z.~Zhang, H.~Yang, J.~Bu, S.~Zhou, P.~Yu, J.~Zhang, M.~Ester, and C.~Wang,
  ``{ANRL}: Attributed network representation learning via deep neural
  networks,'' in \emph{Proceedings of the Twenty-Seventh International Joint
  Conference on Artificial Intelligence, {IJCAI-18}}.\hskip 1em plus 0.5em
  minus 0.4em\relax International Joint Conferences on Artificial Intelligence
  Organization, 7 2018, pp. 3155--3161.

\bibitem{wang2021statistical}
P.~Wang, ``\BIBforeignlanguage{en}{Statistical identification of important
  nodes in biological systems},'' \emph{\BIBforeignlanguage{en}{J Syst Sci
  Complex}}, vol.~34, no.~4, pp. 1454--1470, Aug. 2021.

\bibitem{Kumar2023-ranking}
N.~Kumar and M.~S. Mukhtar, ``\BIBforeignlanguage{en}{Ranking plant network
  nodes based on their centrality measures},''
  \emph{\BIBforeignlanguage{en}{Entropy (Basel)}}, vol.~25, no.~4, Apr. 2023.

\bibitem{zhao2014detecting}
B.~Zhao, J.~Wang, M.~Li, F.-X. Wu, and Y.~Pan, ``Detecting protein complexes
  based on uncertain graph model,'' \emph{IEEE/ACM Transactions on
  Computational Biology and Bioinformatics}, vol.~11, no.~3, pp. 486--497,
  2014.

\bibitem{Rajeh2022-ranking}
S.~Rajeh and H.~Cherifi, ``\BIBforeignlanguage{en}{Ranking influential nodes in
  complex networks with community structure},''
  \emph{\BIBforeignlanguage{en}{PLoS One}}, vol.~17, no.~8, p. e0273610, Aug.
  2022.

\bibitem{Wang_Wang_Zheng_2022mini}
M.~Wang, H.~Wang, and H.~Zheng, ``A mini review of node centrality metrics in
  biological networks,'' \emph{International Journal of Network Dynamics and
  Intelligence}, vol.~1, no.~1, p. 99–110, Dec. 2022.

\end{thebibliography}
\end{document}